\documentclass[a4paper,aps,pra,twocolumn,superscriptaddress,nofootinbib,unpublished,longbibliography
,accepted=2021-01-07
]{quantumarticle}
\pdfoutput=1

\usepackage{footmisc}
\usepackage[hidelinks]{hyperref}
\usepackage{mathtools}
\usepackage{amsmath,amssymb,graphicx,xcolor,braket,hyphenat,makeidx,amsthm}
\usepackage{bbm}
\usepackage{relsize}
\usepackage{comment}
\usepackage{cases}
\usepackage[toc,page]{appendix}

\usepackage{thm-restate}
\usepackage{thmtools,nameref,cleveref}

\newtheorem{definition}{Definition}

\newtheorem{lemma}{Lemma}

\newtheorem{prop}{Proposition}

\usepackage{lipsum}

\usepackage{dsfont}
\usepackage{bbm} 
\usepackage[normalem]{ulem} 
\usepackage{color} 
\usepackage{braket} 

\newcommand{\ketbra}[2]{\ket{#1}\!\!\bra{#2}}
\newcommand{\tr}{\textup{tr}}
\newcommand{\be}{\begin{equation}}
\newcommand{\ee}{\end{equation}}
\newcommand{\nn}{{\mathbbm{N}}}
\newcommand{\nnp}{{\mathbbm{N}_{> 0}}}
\newcommand{\nno}{{\mathbbm{N}_{\geq 0}}}

\newcommand{\rr}{{\mathbbm{R}}}

\newcommand{\cc}{{\mathbbm{C}}}
\newcommand{\zz}{{\mathbbm{Z}}}
\newcommand{\me}{\mathrm{e}}
\newcommand{\mi}{\mathrm{i}}
\newcommand{\id}{{\mathbbm{1}}} 
\newcommand{\idch}{{\mathcal{I}}}  

\newcommand{\lo}{o}

\newcommand{\Mspace}{\vspace{0.18cm}} 

\def\ba#1\ea{\begin{align}#1\end{align}} 

\newcommand{\doc}{\text{paper}}  
\newcommand{\App}{\text{Appendix}}
\newcommand{\app}{\text{appendix}}

\newcommand{\cl}{{\textup{C}}}

\newcommand{\re}{{\textup{R}}}
\newcommand{\reT}{{\textup{R}_{\textup{T}}}}  
\newcommand{\reI}{{\textup{R}_{\textup{I}}}} 

\newcommand{\mpw}[1]{{[\color{blue}#1]}}

\newcommand\mpwS[1]{{\let\helpcmd\sout\parhelp#1\par\relax\relax} }
\long\def\parhelp#1\par#2\relax{%
	\helpcmd{#1}\ifx\relax#2\else\par\parhelp#2\relax\fi%
}

\newcommand{\mpwSC}[1]{{\color{blue}\mpwS{#1}}}

\newcommand{\proj}[1]{\ketbra{#1}{#1}}
\newcommand{\circsum}[2]{\overset{#2}{\underset{n= #1}{\mathlarger{\bigcirc} }}}

\renewcommand{\vec}[1]{\boldsymbol{#1}}

\newcommand*{\cI}{\mathcal{I}}
\newcommand*{\cM}{\mathcal{M}}

\newcommand*{\clock}{clockwork}
\newcommand*{\Clock}{Clockwork}
\newcommand*{\bk}{\textup{ct}}
\newcommand*{\linear}{L} 
\newcommand*{\lineaR}{L} 
\newcommand*{\bounded}{\mathcal{B}} 
\newcommand*{\comP}{\circ} 
\newcommand*{\fnr}[1]{\iffalse #1 \fi}
\newcommand*{\NQ}{N_L}
\newcommand*{\m}{{\NQ}}

\begin{document}
	
\title{
Autonomous Ticking Clocks from Axiomatic Principles}

\author{Mischa P. Woods}
\affiliation{Institute for Theoretical Physics, ETH Zurich, Switzerland}

\begin{abstract}
There are many different types of time keeping devices. We use the phrase \emph{ticking clock} to describe those which | simply put | ``tick'' at approximately regular intervals.
Various important results have been derived for ticking clocks, and more are in the pipeline. It is thus important to understand the underlying models on which these results are founded. The aim of this \doc{} is to introduce a new ticking clock model from axiomatic principles that overcomes concerns in the community about the physicality of the assumptions made in previous models. The ticking clock model in~\cite{woods2018quantum} achieves high accuracy, yet lacks the autonomy of the less accurate model in~\cite{thermoClockErker}. Importantly, the model we introduce here achieves the best of both models: it retains the autonomy of~\cite{thermoClockErker} while allowing for the high accuracies of~\cite{woods2018quantum}. What is more,~\cite{thermoClockErker} is revealed to be a special case of the new ticking clock model.
\end{abstract}

\maketitle

\section{Introduction and basics}\label{sec:intro}

Clocks form part of our everyday lives. Understanding their fundamental limitations is an interesting and rich theoretical problem of study which may yield important design principles
 for improved future clocks. However, results pertinent to the performance of clocks are only of relevance if the theoretical clock models underpinning them capture the relevant properties. Therefore, understanding the clock models which underlie the results about clock performance, is as important as the results themselves.\Mspace

Before discussing our findings and their motivation, let us first describe two different types of time keeping devices in order to set the scene for this work. We coin the phrase \emph{ticking clocks} to refer to one type and call the other \emph{stopwatches}. In the literature, both devices are often simply referred to as ``clocks", e.g.~\cite{SaleckerWigner58,Peres80,OptimalStopwatch,Erker,RaLiRe15,WSO16,thermoClockErker,woods2018quantum,Shishir,Yuxiang,PhysRevA.101.042116,2019arXiv191200033H}, yet it is worth introducing distinct names due to their different character. A stopwatch is a device which measures the elapsed time between two external events. There will be a starting time (e.g.\ the beginning of a race) and a stopping time (e.g.\ when the winner crosses the finish line). The stopwatch will attempt to measure the elapsed time. The earliest types of quantum clocks considered in the literature were of this form~\cite{pauli1,pauli2}. These are also the types of clocks one often considers in a metrology setting, since it is equivalent to measuring a phase. However, the action of measuring the stopwatch disturbs its internal dynamics, thus changing the outcome statistics of later time measurements.\Mspace

On the other hand, one can consider a ticking clock which, roughly speaking, is a device which emits ticks at approximately regular intervals. A typical wall clock is a good classical example. Here one can listen or watch the clock face and will know in real-time when it ticks. Analogously to the above example, we will want our mathematical formulation of the ticking clock to allow continuous observations of whether it has ticked or not; without affecting its internal dynamics. Since there is no such requirement imposed on stopwatches, ticking clocks and stopwatches require very different mathematical formulations.\Mspace

Ticking clocks and stopwatches are also physically very distinct objects. The following two examples illustrate this point quite nicely. Firstly, consider a race and measuring the elapsed time between the winner leaving the starting line and crossing the finishing line with a stopwatch. This task can also be carried out by a ticking clock, at least to an accuracy to within plus or minus the time between two consecutive ticks. However, what about if you arranged to meet a friend at a given location at, say, 13:00h tomorrow? If you were only equipped with a stopwatch with no other time reference, you would hopelessly fail to be on time. The reason for this negative predicament, is that you would have no external signal (like the winner crossing the finish line in the previous example) to know \emph{when} to stop your stopwatch | when you eventually press the ``stop button", it may indicate that only 1 second has passed or maybe one week. One may hope to remedy this predicament by resetting their stopwatch immediately after it was stopped; and trying again while keeping a record of the previous outcome. However, this would only lead to a finite number of completely irregular instances when you would know what the time was. Consequently, you would almost surely be very late for the meeting with your friend.\Mspace

The above hypothetical example involving the stopwatch, while conveying an important point, is a bit far fetched from our everyday experience since we do, in fact, always have access to ticking clocks | albeit bad ones | such as the visual difference between day and night. To study such scenarios, one could investigate a different type of time keeping device formed by combining a stopwatch and a ticking clock to take advantage of the best properties of both. Atomic clocks are a good example of such devices. We will leave their study to future work.\Mspace

If either the stopwatch or ticking clock is quantum mechanical in nature, then from a mathematical perspective, both of these devices output information from the clock on Hilbert space $\mathcal{H}_\cl$ to the ``outside". {In the case of a ticking clock, it is advantageous to make} this information transfer explicit within the model by including a register. {At a conceptual level, the temporal information stored in the register is accessible via measurements which, ideally, should not disturb the dynamics of the \clock. This latter property is important since it safe-guards against inadvertently hiding potential disturbances to the \clock{} in measurements on the register.} In the case of a stopwatch, this information retrieval from the clock via a POVM is passive, i.e.\ its retrieval is triggered by an \emph{external} signal, and the clock reacts passively to the measurement~\cite{OptimalStopwatch}. This is in contrast to the information transfer to the register in the case of a ticking clock, in which the information transfer is triggered \emph{internally} by the clock mechanism itself, with no help from external triggers. Both stopwatches and ticking clocks can be modelled by multipartite Hilbert spaces. The most two common elements we will discuss concern the bipartition $\mathcal{H}_\cl\otimes\mathcal{H}_\reT$. In keeping with the terminology of~\cite{RaLiRe15}, we will refer to the space $\mathcal{H}_\cl$ as the \emph{\clock} while $\mathcal{H}_\reT$ will be called the \emph{register}. In both cases, one can define a one-parameter channel $\cM^{t}_{\cl\to\cl}:\,\linear(\mathcal{H}_\cl)\to \linear(\mathcal{H}_\cl)$, where $t\in\mathcal{S}_\bk$ is \emph{coordinate time}. {The exact nature of the register depends heavily on the particulars of the model, and will become less opaque in the coming sections.}\Mspace

In these models, one should think of coordinate time as some unknown parameter which increases as time advances. It is used as a bookkeeping parameter, i.e.\ it is assumed that the channels $\cM^{t}_{\cl\to\cl}(\rho_\cl)$, $\cM^{\prime\, t}_{\cl\to\cl}(\rho_\cl')$ of two distinct stopwatches or ticking clocks correspond to the state of the clocks at the same ``time''. In these models we assume that {the ticking clock or stopwatch is initiated at a particular time,} (i.e.\ that there is a minimum $t$ for which the channel is defined for; which w.l.o.g., we can set to zero). The motivation is that, as we will see, ticking clocks emit temporal information to the outside in an irreversible fashion. 
Furthermore, time may be fundamentally continuous or discrete. These two separate cases are conveniently modelled by defining $\cM^{t}_{\cl\to\cl}$ for $t\in\mathcal{S}_\bk=[0,\infty)$ or $t\in \mathcal{S}_\bk=(0,\delta, 2\delta,\ldots)$ respectively. Here $\delta>0$ is some fixed parameter which allows (if desired) to define a continuous time clock from a discrete one, by taking the limit $\delta\rightarrow0^+$ in an appropriate way. These two cases will be referred to as discrete coordinate time and coordinate time respectively. Finally, observe that we should not think of the coordinate time as being physical, in the sense that we could define a new coordinate time through a change of variable $t':=f(t)$ for some strictly increasing function $f:\rr\to\rr$ so long as the ticking clocks when parametrised by $t'$ rather than $t$ satisfy the to-be-defined in \cref{sec:new quantum clock def}, axiomatic definition of a ticking clock.\Mspace

For concreteness, the rest of this \doc{} will concern the nature of ticking clocks. Conceptually, the goal of the \clock{} is to provide the timing; changing the state of the register at the right moments | analogously to how the clockwork in a wall clock is the mechanism which moves the clock hands to produce ticks. As such, it should not need any timing from the ``outside''. This physical requirement has been captured mathematically in previous models~\cite{RaLiRe15,woods2018quantum} by requiring that the \clock{} be Markovian (also know as a divisible channel), meaning
\begin{align}\label{eq:divisible}
\cM_{\cl\to\cl}^{t_1+t_2}= \cM_{\cl\to\cl}^{t_1} \circ \cM_{\cl\to\cl}^{t_2},
\end{align}
for all $t_1,t_2\in  \mathcal{S}_\bk$. 
This condition has been justified by considering the opposite scenario: suppose that $\cM_{\cl\to\cl}^t$ were not divisible, i.e.\ \cref{eq:divisible} does not hold for some $t_1,t_2\in  \mathcal{S}_\bk$. 
Then, the channel being applied per unit of coordinate time (discrete or continuous), would have to depend on knowledge of the value of coordinate time itself. In other words, the device may need an additional time reference external to the setup. However, by definition the \clock{} is supposed to contain all sources of timing necessary for the ticking click to function. Requirement \cref{eq:divisible} {can be verified using the techniques discussed in \cite{Rivas_2014}.} This equation is discussed further in \cref{sec:Self-timing}.\Mspace

Demanding \cref{eq:divisible} has some immediate consequences, the most important of which is that
the \clock{} is fully determined at all times by the ``smallest coordinate time step''. In the case of continuous coordinate time, under appropriate continuity assumptions, this reads:
\begin{align} \label{eq_mapDeltadelta}
\cM^{t}_{\cl \to \cl} = \lim_{\substack{\vspace{0.1cm} \delta \to 0^+ \vspace{0.05cm}\\ t/\delta\in\zz }} \bigl(\cM^{\delta}_{\cl \to \cl}\bigr)^{\comP \frac{t}{\delta} }
\end{align}
for all $t\in\mathcal{S}_\bk$. 
In the case of discrete coordinate time, \cref{eq_mapDeltadelta} holds if one does not take the limit $\delta \to 0^+$. The $\comP$ in the power in \cref{eq_mapDeltadelta} represents composition of the channel $\cM^{\delta}_{\cl \to \cl}$ with itself $t/\delta$ times. We will use this notation thought this \doc.
The authors of~\cite{RaLiRe15} construct a ticking clock model by describing how the \clock{} they introduce interacts with a register.\Mspace

\emph{\bf Paper overview:} We review the model~\cite{RaLiRe15} in \cref{sec:the ATG model} followed by discussing its downsides. Then in \cref{sec:basic principles} we discuss two important principles for ticking clocks, followed by the axiomatic definition of a new ticking clock model in \cref{sec:new quantum clock def}. This new model satisfies the conditions introduced in \cref{sec:basic principles}, while overcoming the shortcomings discussed in \cref{sec:the ATG model}. We then discuss some important properties and definitions for the new ticking clock before moving on to \cref{sec:Explicit Ticking Clock Representation}, where we formulate a channel representation for the axiomatically defined ticking clock. We then show how the ticking clock representation admits an autonomous implementation, and prove some properties for its \clock{}. In \cref{sec:Ticking Clock Examples} we show how previous examples of clocks from the literature are either special cases of the new formulation or can easily be adapted to fit into it. In \cref{sec:accuracy} we discuss measures of accuracy for the clock and conclude with a summary of the highlights and outlook in \cref{sec:discussions}.

\section{The ticking clock model of~\cite{RaLiRe15}}\label{sec:the ATG model}

\subsection{Model description}

The authors start by describing their model for a discrete coordinate time ticking clock.\footnote{The authors of~\cite{RaLiRe15} do not use the terminology ``ticking clocks'' and instead refer to their devices as ``clocks''. In keeping with the terminology introduced in this \doc, we will use the former denomination.} The continuous coordinate time ticking clock, is then realised by allowing the discrete time step parameter $\delta$ to tend to zero while demanding a certain continuity condition.\Mspace 
 
We start with their notion of the tick register with Hilbert space $\mathcal{H}_\reT$ for a ticking clock. This is a memory to which the temporal information coming from the \clock{} is recorded. The register is formed by a tensor product space, $\mathcal{H}_\reT=\mathcal{H}_{\re_1}\otimes\mathcal{H}_{\re_2}\otimes\mathcal{H}_{\re_3}\otimes\ldots$ over local registers $\mathcal{H}_{\re_i}$ which are all isomorphic to a fixed register $\mathcal{H}_\reI$.\Mspace 

The authors define a channel $\cM^{\delta}_{\cl\to\cl\reI}:\lineaR(\mathcal{H}_\cl)\to\lineaR(\mathcal{H}_\cl\otimes\mathcal{H}_\reI)$ for some fixed $\delta>0$. This gives rise to the state of the register after $N$ applications of the channel ($N$ discrete coordinate time steps); denoted
$\rho_\reT(N)=\rho_{\re_1\re_2\re_3\dotsb}$. Its local states are given by
\begin{align}\label{eq: rhp R i}
\rho_{\re_l}:= \begin{cases}
\tr_\cl\left[ \cM^{l \delta}_{\cl\to\cl\reI}\left( \rho_\cl^0 \right) \right], &\mbox{for } l=1,2,\ldots,N,\\
\rho_{\re_l}^0, &\mbox{for } l=N+1,\ldots\end{cases}
\end{align}
where $\rho_{\re_l}^0$ is the $l^\text{th}$ initial local state, and $\cM^{l\delta}_{\cl\to\cl\reI}$ is defined recursively by applying the channel of the \clock{} $l$ times:  $\cM^{l\delta}_{\cl\to\cl\reI} :=\cM^{\delta}_{\cl\to\cl\reI}\left( \tr_{\reI}\left[\cM^{(l-1)\delta}_{\cl\to\cl\reI} \left(\rho_\cl^0\right)\right]\right)$, $l\in\nnp$.\Mspace

As discussed in \cref{sec:intro}, every application of the channel $\cM^{\delta}_{\cl\to\cl\reI}$ needed in the construction of \cref{eq: rhp R i} corresponds to one time step of discrete coordinate time. {However, \cref{eq: rhp R i} does not yet constitute the ticking clock model in \cite{RaLiRe15}; it requires one more ingredient | a so-called \emph{gear system}. It has a corresponding channel, called the gear system channel, denoted $G_{\reT\to\reT}$, which moves every local register site to the left by one local site in-between every application of the \clock{} channel; see \cref{Fig:new fig}.\Mspace

As depicted in \cref{Fig:new fig}, we can think of the \clock{} initially interacting with the 1st register, with all the other registers to the right. It is convenient to only keep track of the register which the \clock{} is currently acting on and those which are to its right. Keeping to this convention, the gear system channel $G_{\reT\to\reT}$ when applied to a product register state 
\begin{align}\label{eq:prod state}
\rho_\reT=\rho_{\re_{1}}^{(1)}\otimes\rho_{\re_2}^{(2)}\otimes\sigma_{\re_3}^{(3)}\otimes\ldots
\end{align}
$m$ times achieves $G_{\reT\to\reT}^{\,\comP m}(\rho_\reT)=\rho_{\re_{1+m}}^{(1+m)}\otimes\rho_{\re_{2+m}}^{(2+m)}\otimes\rho_{\re_{3+m}}^{(3+m)}\otimes\ldots$.\Mspace

The gear system is an integral part of the ticking clock, since without it, the \clock{} would not be able to access all the local registers; see \cref{Fig:new fig}.} It may be a mechanical system such as a rack and pinion, or non mechanical such as a kinetic degree of freedom associated with the register, turning it into flying qubits on a line in rectilinear motion.\Mspace 

Thus the gear system moves the register along as if it were on a conveyor belt in the following sense: for the 1st application of the channel $ \cM^{\delta}_{\cl\to\cl\reI}$, the \clock{} interacts with the 1st register $\mathcal{H}_{\re_{1}}$ and the registers are then instantaneously moved to the left by one local register site via the gear system so that the \clock{} now interacts with the 2nd register $\mathcal{H}_{\re_2}$ for the second application of the channel $\cM^{\delta}_{\cl\to\cl\reI}$. The gear system then moves all the register sites by one site to the left as before and the process is repeated indefinitely. The local states of the register $\tr_\cl\left[ \cM^{l \delta}_{\cl\to\cl\reI}\left( \rho_\cl^0 \right) \right]$, generated via the  $l^\textup{th}$ application of the channel, determine whether a tick occurred or not. The intuition is that the \clock{} is releasing some temporal information at every application of the channel, so that after a sufficiently large number of applications of the channel it will contain too little temporal information to be useful and the register states $\rho_{\re_l}$ for sufficiently large $l$, contain very little temporal information. The registers could contain tick/no-tick information by having the channel write ``0'' to the $l^\text{th}$ register, $\tr_\cl\left[ \cM^{l \delta}_{\cl\to\cl\reI}\left( \rho_\cl^0 \right) \right]=\ketbra{0}{0}_\reI$, in the case of no-tick, or a ``1'' in the case of a tick, $\tr_\cl\left[ \cM^{l \delta}_{\cl\to\cl\reI}\left( \rho_\cl^0 \right) \right]=\ketbra{1}{1}_\reI$. Here $\ket{0},\ket{1}$ are two orthogonal states.\Mspace

The authors then introduce a continuous coordinate ticking clock by demanding that the channel $\tr_{\reI}\left[\cM^{\delta}_{\cl\to\cl\reI}\right]$ satisfies an $\epsilon$-continuity condition. {That is to say, there exists $\epsilon:\rr\to\rr$ such that,}
\begin{align}
\left\| \tr_{\reI}\left[\cM^{\delta}_{\cl\to\cl\reI}\right] - \idch_\cl \right\|_\diamond \leq \epsilon(\delta),
\end{align}
where $\epsilon(\delta)\rightarrow 0^+$ as $\delta\rightarrow 0^+$. $\idch_\cl$ is the identity channel on the \clock{}, and $\| \cdot\|_\diamond$ is the diamond norm. The authors then specify that one applies the channel $\cM^{\delta}_{\cl\to\cl\reI}$ a number of times which is proportional to $\delta \epsilon$ with $\delta$ of order $1/\epsilon$, so as to achieve non-trivial ticking clock dynamics. The continuum limit case was further studied in~\cite{woods2018quantum} with a few additional physically motivated constraints introduced. See \cref{Fig:gearSystemClock}~a) for a depiction of the combined \clock{} and gear system channels.\Mspace

\begin{figure}\includegraphics[scale=0.5]{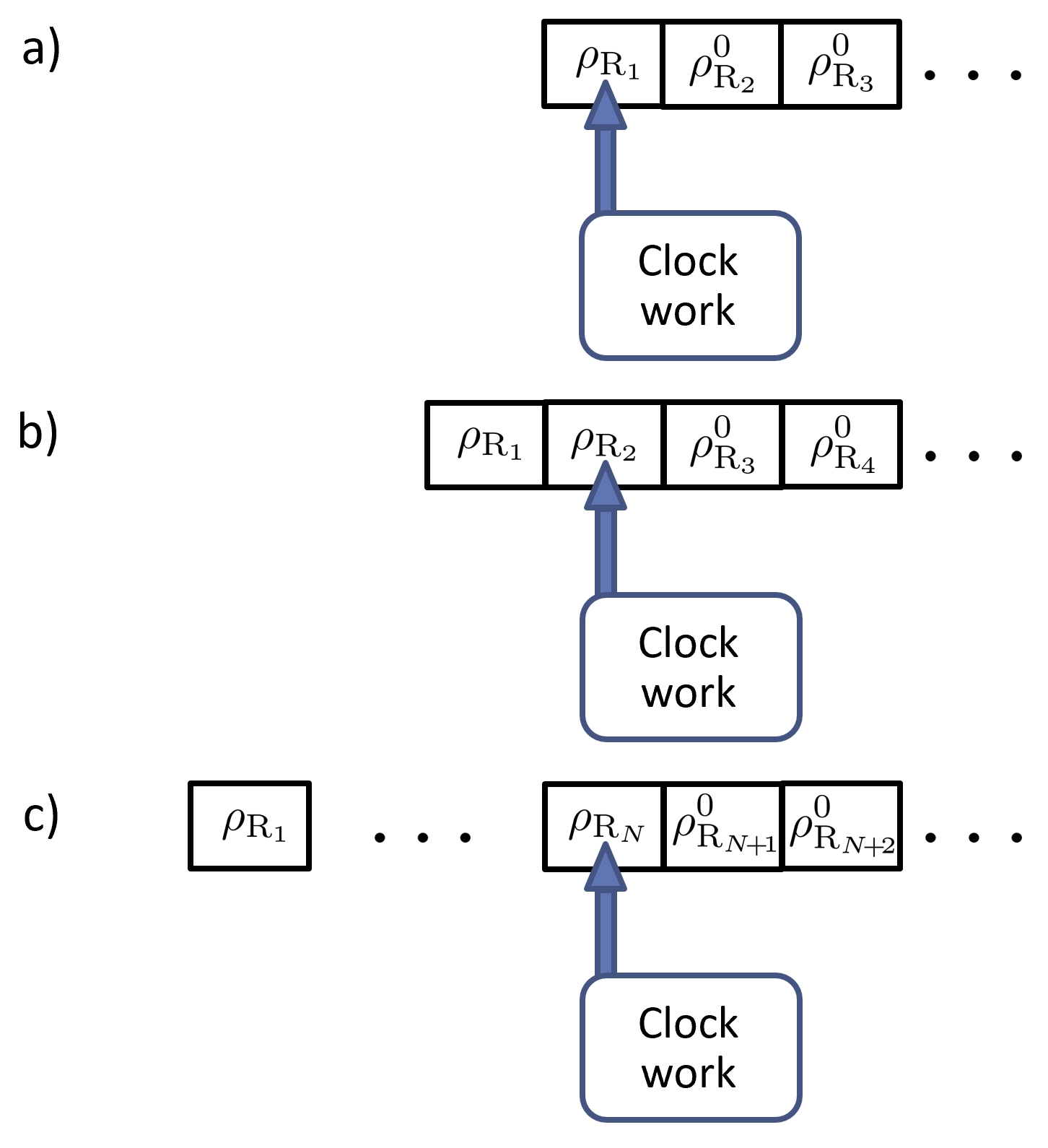}
	\caption{{In figures a) to c), square boxes indicate the local register sites while the \clock{} acting on a local register site is indicated via a blue arrow.\\ {\bf a)} Depiction of the ticking clock for $t=\delta$.   The clockwork channel is applied once and the output on the register, $\rho_{\re_1}$, written to the 1st register site.\\
	{\bf b)} Depiction of the ticking clock for $t=2\delta$. Observe that the register has been shifted by one local register site to the left via an application of the gear system channel and $\rho^0_{\re_2}$ has been transformed into $\rho_{\re_2}$ via one application of the \clock{} channel.\\
	{\bf c)} Depiction of the ticking clock for $t=N\delta$. Observe that the register has been shifted by $N$ local register sites when compared with a), the first $N$ local register sites have been written to and the \clock{} channel has been applied $N$ times in total.
	}}\label{Fig:new fig}
\end{figure}

\begin{figure}\includegraphics[scale=0.456]{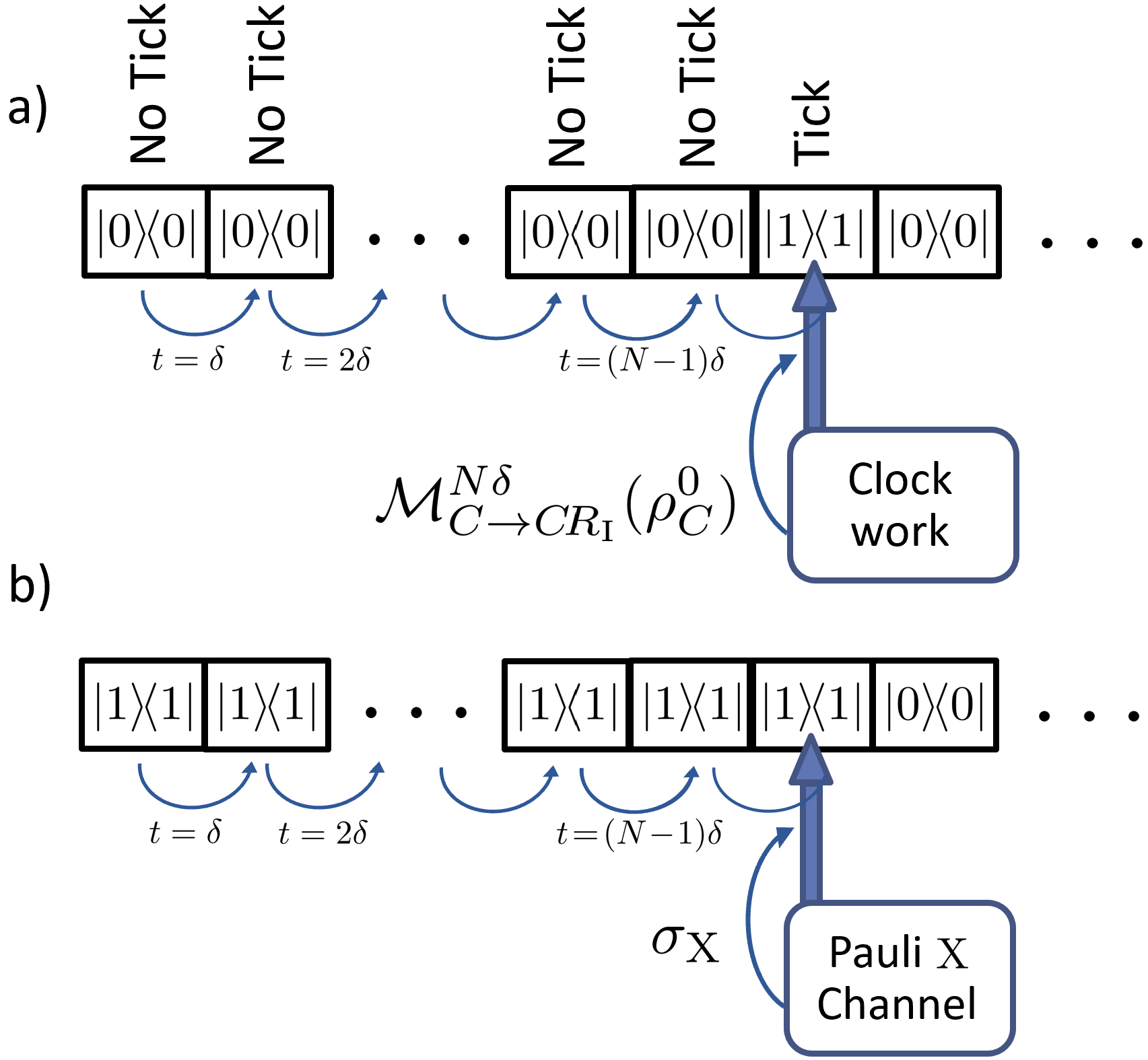}
	\caption{{\bf a)} Illustration of the ticking clock model in~\cite{RaLiRe15}: All the local register site qubits $\sigma_{\re_1},\sigma_{\re_2}\ldots$ are initially set to $\proj{0}$. The channel $\tilde \cM^{\delta}_{\cl\reT\to\reT}:=G_{\reT\to\reT}\!\left( \tr_\cl\left[\cM^{\delta}_{\cl\to\cl\reI}\left(\rho_\cl^0\right)\right]\tr_\reI[\,\cdot\,]\otimes\cI_{\reI}(\cdot)\otimes\cI_{\reI}(\cdot)\otimes\ldots \right)$, where $\cI_{\reI}(\cdot)$ is the identity channel, is then applied repeatedly at times $t=\delta,2\delta,3\delta,\ldots$. The state of the register at some fixed time $t>0$ is obtained by setting $N=t/\delta$ and taking limit $\delta\rightarrow 0^+$. The number of 1's corresponds to the number of ticks which have occurred in time interval $[0,t]$.\\ {\bf b)} Same scenario as in a) but now swapping the \clock{} channel $\tr_\cl\left[\cM^{\delta}_{\cl\to\cl\reI}\left(\rho_\cl^0\right)\right]\tr_\reI[\,\cdot\,]$ with the Pauli X channel $\sigma_\textup{X}(\cdot)$ which maps $\proj{0}$ to $\proj{1}$. The register now records the time with zero error, even though the Pauli X channel produces no temporal information | unlike the channel it replaced. We thus see that all the temporal information comes solely from the gear system $G_{\reT\to\reT}$. Alternatively, the same observation holds when using $\proj{1}_\reI\tr_\reI[\cdot]$, rather than $\sigma_\textup{X}(\cdot)$. Even in scenario a), the gear system is functioning as a perfect stopwatch: by counting the number of zeros between ticks, one can determine precisely the coordinate time interval between ticks.}\label{Fig:gearSystemClock}
\end{figure}

\subsection{Drawbacks}


{Unlike the \clock{} channel, for which the resources required for its implementation such as energy or dimensionality are studied, the gear channel is always considered to be a free resource whose implementation is not studied.}
However, if such a gear system channel existed, rather than applying it {in tandem with} the \clock{} channel $ \cM^{\delta}_{\cl\to\cl\reI}$, one could arguably apply it {in tandem with} a much simpler channel | bypassing the \clock{} altogether | and achieve an \emph{idealised ticking clock}. In this instance, this means a ticking clock for which $\tr[\rho_\reT(t) \rho_\reT(t')]=\hat {\delta}(t-t')$ for all $t,t'\in  \mathcal{S}_\bk$, where $\hat \delta(\cdot)$ is the Dirac-delta function in the continuous coordinate time limit, and a Kronecker-delta in the discrete coordinate time case. {Physically speaking, this means that the state of the register at any two different coordinate times are orthogonal to each other. Consequently, the coordinate time at any given instance can be determined exactly from the registers via measurement.} {See \cref{Fig:gearSystemClock}~b) for an illustration of this drawback.}\Mspace

 This highlights the 1st drawback with such a model: the gear system channel | while it is not supposed to contain temporal information | actually can function like an idealised ticking clock. While this argument is rather indirect, the following argument is direct in the sense that it applies to all ticking clocks of the authors using the gear system in conjunction with the \clock{} as intended.\Mspace
 
If one has a ticking clock, and no other resource, then they should arguably \emph{not} be able to determine the precise time between ticks | to do that, one would need an additional time keeping device, such as a  very precise stopwatch. However, by simply counting the number of zeros between the ones in the register, one can determine precisely the time between ticks. This is a direct consequence of the gear system moving the register along by one qubit sequentially in perfect tandem with the passing of coordinate time. Observe also that this holds independently of how regular the ticks are | it could be a ticking clock which is very accurate and the ticks occur at highly regular intervals, or very imprecise with ticks occurring randomly with respect to coordinate time. {It is also important in this argument that one can measure all the local registers to access all the zeros and ones. This is an assumption in their model. In the continuum limit, this would require accessing an infinite amount of information. How one would do so, is another open question.}\Mspace

One might hope to remedy these {drawbacks} by simply removing the gear system {altogether} and allowing the \clock{} to always write to the same initial qubit register at all times. While this indeed means that the number of zeros (``no-ticks'') emitted by the \clock{} between the ones (``ticks'') is now not recorded in the register as one would like, the state of the register in the instances when ticks occur will now be the same regardless of how may times the clock has ticked. Likewise, the register will also be in the same state (\,$\proj{0}_\reI$) between any two consecutive ticks. As such, in the continuous time limit, the register will be in the state $\proj{1}_\reI$ a measure zero amount of time, and in the state $\proj{0}_\reI$ almost always. This behaviour is clearly problematic and at odds with that of familiar ticking clocks, such as a wall clock.\Mspace

Ideally, one satisfactory option would be to assign just one qubit of memory to which all the zeros (``no ticks'') information is written to, and a new register qubit to be allocated to the output of the \clock{} every time the ticking clock ticks. Therefore, at any instance, one would be able to determine the time by simply reading the number of ones in the register, but since the no tick information is always overwritten, one would not be able to determine how much coordinate time has passed between ticks. This hypothetical solution is unfortunately not possible, since it would require the gear system channel $G_{\reT\to\reT}$ to know when the \clock{} is going to tick, yet since it acts solely on the register, it cannot do so. This highlights the difficulty of removing the gear system or altering its behaviour in a beneficial way. In \cref{sec:new quantum clock def}, we will show, via explicit construction, a satisfactory solution. \Mspace

The above described drawback holds true for both the discrete and continuous time coordinate ticking clocks. In the continuous time coordinate case, there is also another drawback related to its memory storage. We will leave its discussion to \cref{sec:Finite memory}.\Mspace

Note that the authors do suggest that a high precision gear system is not needed.
Their argument is based on assuming that the gear system can fail with some probability $p$, where fail means that the gear channel $G_{\reT\to\reT}$, is replaced with the channel $\tilde G_{\reT\to\reT}=p\cI_\reT+(1-p)G_{\reT\to\reT}$, with $\cI_\reT$ the identity channel on $\mathcal{H}_\reT$. However, their reasoning is based on the fact that replacing the channel $G_{\reT\to\reT}$ with this one, incurs (at most) an irrelevant change in the state of the \clock{} at arbitrary coordinate times $t\in \mathcal{S}_\bk=(0,\delta, 2\delta,\ldots)$. However, there are several issues with this approach. On the one hand, no study of the induced change in the register is produced; yet the accuracy of the ticking clocks according to their measure (the Alternative Ticks Game), is solely a function of the register states on $\mathcal{H}_\reT$ in the large coordinate time limit. Second, in the continuous ticking clock limit ($\delta\to 0^{+}$), the gear system moves the register continuously, and any physically motivated gear system may produce errors which are irreconcilable with the error model described above. Furthermore, errors in the gear system are {cumulative}, and if the gear system writes the tick to an incorrect location in the register, it is possible to change the outcome of the alternative ticks game, thus changing its accuracy according to this measure.\Mspace

\begin{figure}[h]
	\includegraphics[scale=0.35]{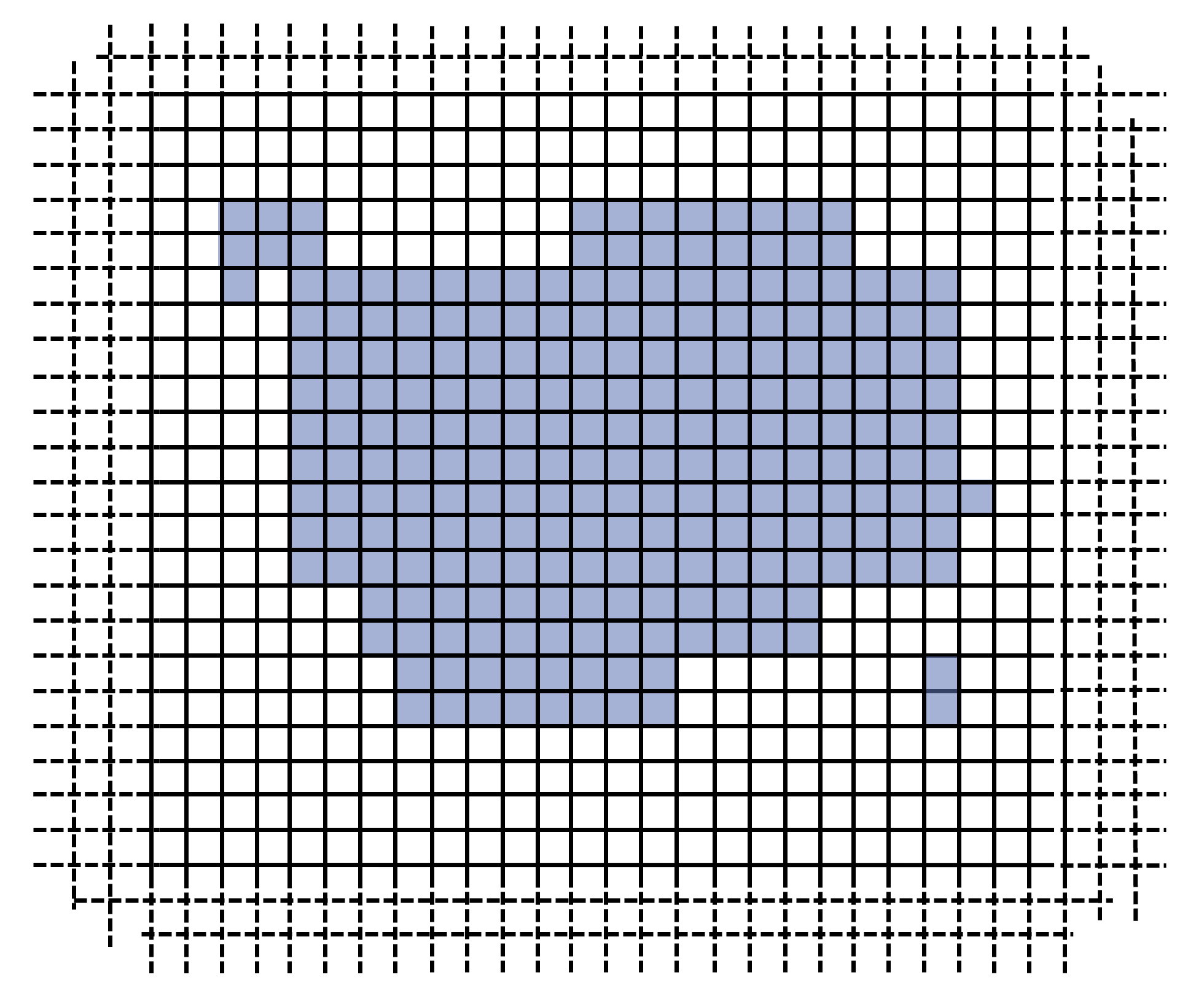}
	\caption{A depiction of a clock's register satisfying the finite running memory condition. The grid represents {a section of} the (possibly infinite dimensional) space $\mathcal{H}_\reT$. {The dotted lines indicate that the grid may continue ad infinitum, while each square in the grid represents one qubit of memory.} If the condition \cref{eq:finite memory condition} can be satisfied, then one can find a subspace {in which $\hat P$ projects onto (the blue squares in the figure)} which includes all | up to an arbitrarily small amount $\epsilon$ | of the changes produced in the register during the time interval $[0,t]$, such that $\hat P_\bot \rho_\reT(0)\hat P_\bot \approx \hat P_\bot \rho_\reT(t)\hat P_\bot$, i.e. that there has been effectively no change in the rest of the register (depicted by the white squares).}\label{Fig:finite memory condition}
\end{figure}

\section{Two basic principles for ticking clocks models: Finite running memory and Self-timing}\label{sec:basic principles}
We now present two basic principles to physically motivate descriptions of ticking clocks. The first is conceptually desirable, but arguably not necessary, while the second is more essential. 

\subsection{Self-timing}\label{sec:Self-timing}
Understanding the underlying timing resources of a ticking clock is an important task. Otherwise, any physical implementation of it may require unaccounted for timing resources. Therefore identifying and quantifying such resources is important. 
A simple counter-example where the timing resources are unaccounted for, is a clock model with unitary dynamics governed by a time-dependent Hamiltonian over the \clock{} and register.
\Mspace 

One should distinguish the concept of self-timing from that of autonomy. An autonomous ticking clock can be thought of as one in which all resources for the clock to run can be explicitly accounted for. An example of an autonomous clock is~\cite{thermoClockErker}. 
Such ticking clocks are clearly also self-timing but the contrary is not necessarily true. The extent to which the ticking clock model presented in this \doc{} is autonomous, will be discussed in \cref{sec:Explicit Ticking Clock Representation,Sec:conclusions}.\Mspace 

We say that the (continuous coordinate time) ticking clock is \emph{self-timing} if 
its  
one-parameter channel on the \clock{} and register,  $\cM^{t}_{\cl\reT\to\cl\reT}$, 
is divisible: 
\begin{align}
\cM^{t_1+t_2}_{\cl\reT\to\cl\reT}=\cM^{t_1}_{\cl\reT\to\cl\reT} \circ \cM^{t_2}_{\cl\reT\to\cl\reT}\label{eqdef:self containment}
\end{align}
for all $t_1,t_2\geq 0$. The reasoning is that if this  {were} not the case, then one could use systems alien to the register and \clock{} to provide timing. In fact, after the identification of the register space $\mathcal{H}_\reT$, the smallest additional space ones needs to include so that \cref{eqdef:self containment} is satisfied (if such a space exists), is a means with which to identify a \clock{} space.\Mspace

One may wonder why we do not demand this divisibility requirement directly for the \clock{} channel, since indeed, the point of the \clock{} is to provide all the timing | the register should be a passive element. However, we will see in \cref{prop:clockwork channel} that while the \clock{} will indeed be divisible in most circumstances, there will be others in which it may not, yet the register will not be providing a source of timing in these cases.\Mspace 

This definition of self-timing differs from that of self-containment from~\cite{RaLiRe15}. In particular, the output in \cref{eqdef:self containment} is on the entire register $\reT$. 
Therefore the self-contained ticking clocks from~\cite{RaLiRe15} are not necessarily self-timing according to the above definition. One can of course make them self-timing by including explicitly the channel for the gear system together with that of the \clock; however, the estimates of the ticking clock's precision in~\cite{RaLiRe15,woods2018quantum} are solely based on {the properties of the \clock{} alone, and thus do not take into account the properties of the gear system}, yet the gear system itself can be used as an idealised clock (as discussed in \cref{sec:the ATG model}).

\begin{figure}[h]
	\includegraphics[scale=0.55]{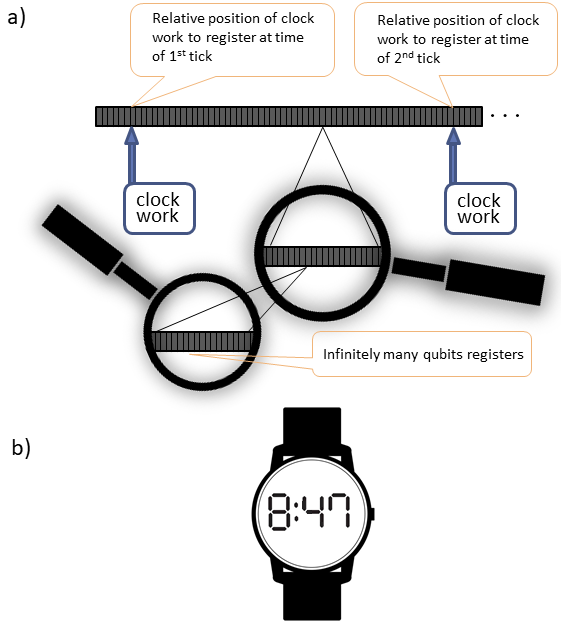}
	\caption{{\bf a)} Register from~\cite{RaLiRe15}: in the continuous coordinate time limit ($\delta\to 0$), during any finite time interval (e.g.\ between two consecutive ticks), the \clock{} needs to interact with an infinite number of qubit registers $\reI$ | one at a time and sequentially. The blue triangle indicates the location of the \clock{} relative to the register site it is writing to. The grey bar represents the infinite number of register sites.\\ {\bf b)} Depiction of a digital wrist watch. Here between seconds (or ``ticks'') the digital display does not change, and during any finite coordinate time interval, the display exhibits a finite number of distinguishable states. This is analogous to the register in the new model presented here: between ticks, the register does not change.}\label{Fig:memory}
\end{figure}

\subsection{Finite running memory}\label{sec:Finite memory}
A requirement for any realistic model of a ticking clock is that it only utilises finite resources per unit of coordinate time. In this section we introduce a definition which captures this notion for the clock's register by demanding that the \clock{} can only invoke a finite change on it per unit of coordinate time.\Mspace

We say that a ticking clock requires \emph{finite running memory} if for every tuple $\big(\epsilon\!>\!0$, $t\!>\!0$, $\rho_\reT(0)\in\mathcal{S}\left(\mathcal{H}_\reT\right)\big)$ there exists a projector $\hat P$ onto a \emph{finite} dimensional subspace $\mathcal{H}_P\subseteq \mathcal{H}_\reT$ such that 
\begin{align}\label{eq:finite memory condition}
\begin{split}
&\big\|  \rho_\reT(t)- \hat P \rho_\reT(t)\hat P -(\hat P_\bot \rho_\reT(t)\hat P+h.c.)\\
&\,\, -\hat P_\bot \rho_\reT(0)\hat P_\bot  \big\|_1\leq \epsilon,
\end{split}
\end{align}
where $\rho_\reT(t):=\tr_\cl\left[ \cM^{t}_{\cl\reT\to\cl\reT}(\rho_{\cl\reT})\right]$ is the state of the register at coordinate time $t$, and $\hat P_\bot:=\id_\reT-\hat P$.  See \cref{Fig:finite memory condition} for a graphical illustration. Note that it is important that the condition holds for all initial register states on $\reT$, even if some initial register states are not relevant for the functioning of the ticking clock. This is because such states are physical and if the ticking clock is a physically realistic model, it should satisfy the finite running memory requirement even in these scenarios | regardless of whether the register is correctly encoding the information from the \clock{} in such cases. This reasoning is analogous to why quantum information theorists demand quantum channels be completely positive rather than just positive | even if they do not intend to apply their channels on entangled states.\Mspace

Observe that the register in the ticking clock model of~\cite{RaLiRe15} does not satisfy the finite running memory requirement since the output on the register is independent of the initial register state, and the \clock{} has to interact with infinitely many copies of the register subspace $\reI$ in any finite proper interval of coordinate time; see \cref{Fig:memory}~a).\Mspace


One may feel that infinite dimensional registers are physical since, indeed, spaces with continuous spectrum are physical. Consider for example the case in which the register is ``a particle in a box'' $\mathcal{H}_\reT=L^2[0,1]$. 
One could in principle store an infinite amount of information in the box by partitioning it into infinitesimally small orthogonal compartments. However, due to technological constraints, such information would not be retrievable nor writable, and a more realistic setup would be to store only a finite amount of information in finitely many partitions | each one, containing an infinite number of orthogonal states. Any resolvable reader would then consist of a projective measure $\{\hat P_l\}_{l\in\nn}$
, where each $\hat P_l$ projects onto one compartment of the register. This way, while each $\hat P_l$ may project onto an infinite dimensional subspace, one can never discern between different orthogonal states on the subspace. Under such a condition, the finite running memory condition given by \cref{eq:finite memory condition} should still hold when $\hat P$ is replaced with any linear combination of a \emph{finite} number of projectors $\hat P_l$. What is more, we {would} also require that the ticking clock cannot \emph{write} an infinite amount of information to every register subspace. One way to ensure this, would be to require that the ticking clock channel written in Kraus {form}, $\cM_{\cl\reT  \to \cl  \reT}^t(\cdot)=\sum_{n} K_n(t) (\cdot) K_n^\dag(t)$, has Kraus operators which admit an expression $K_n(t)=\sum_{l\in\nn}{\hat\gamma_{l,n}(t)\hat P_l}$
, for some operators $\hat\gamma_{l,n}(t)\in\bounded(\mathcal{H}_{\cl\reT})$.\Mspace

For simplicity, the model we introduce in the following section will satisfy the former finite running memory condition; \cref{eq:finite memory condition}.
This is to say, the projectors $\hat P_l$ will project onto finite dimensional spaces. It could however be generalised to contain a register satisfying the latter condition also.

\section{New ticking clock model}\label{sec:new quantum clock def}

We now propose a ticking clock model through a set of physically motivated axioms. It will be self-timing and of finite running memory. Unlike the model discussed in \cref{sec:the ATG model}, it will be a continuous coordinate time model from the outset. 
We discuss its accuracy in \cref{sec:accuracy}.\Mspace


The following describes the extension of a \clock{} channel $\cM^{t}_{\cl \to \cl }$ 
to include the interaction with the register $\reT$. All the conditions regarding how the ticking clock functions will be laid-out in this section. While some of these will be similar to those of \cref{sec:the ATG model}, the new model will not assume any of the conditions nor setup from said section. For example, it will not require a gear system.\Mspace

The tick register here is also different to that described in \cref{sec:the ATG model}. It consists of $N_T+1$ orthonormal states $( \ket{0}_\reT, \ket{1}_\reT,\ket{2}_\reT,\ldots, \ket{N_T}_\reT )$ representing no tick, 1 tick, 2 ticks, \ldots, $N_T\in\nnp$ ticks respectively. While it is clear that any ticking clock with a finite dimensional tick register satisfies the finite running memory condition of \cref{sec:Finite memory}; in this case, the fulfilment of this condition is not inherently related to its finite dimensionality. Indeed, if one takes the infinite dimensional limit $N_T\to\infty$ in the ticking clock model in \cref{sec:Explicit Ticking Clock Representation} which results from the axioms of the current section, the resulting ticking clock \emph{also} satisfies the finite running memory condition of \cref{sec:Finite memory}.\Mspace

To start with, we describe a \emph{periodic} register which resembles the familiar clock which repeats itself whenever the memory is full, e.g.\ every 12 or 24 hours. It naturally satisfies  $\ket{n}_\reT=\ket{n \textup{ mod. } N_T+1}_\reT$, for $n\in\zz$. Later in this section we will consider a variant of this.\Mspace

\fnr{In the following we will only refer to the periodic and cut-off register cases since the infinite memory case can be generated from either register model by taking the $N_T\to\infty$ limit.}

In order for a device to be considered a ticking clock, it should satisfy some conditions on its \clock{} and tick registers. After introducing the following shorthand notation, we discuss 5 such conditions. {Let}
\begin{align}\label{eq:k channel notation def}
\cM^{t,k}_{\cl \to \cl \reT}(\cdot):=&\,\cM^{t}_{\cl \reT \to \cl \reT}\big((\cdot)\otimes \proj{k}_\reT\big):\\
& \mathcal{B}\left(\mathcal{H}_\cl\right) \mapsto \mathcal{B}\left(\mathcal{H}_\cl \otimes\mathcal{H}_\reT\right),\nonumber
\end{align}
{for} $k=0,1,\ldots,N_T$ denote the ticking clock channel with a finite dimensional \clock{} when acting on the $k^\text{th}$ register state.\Mspace

\emph{\textbf{\textup{1)}} Time invariance symmetry condition}:\footnote{In the following and throughout this \doc, we will often omit tensor products with the identity operator on either $\cl$ or $\reT$ for brevity.}
\begin{align}\label{eq:clock invariance}
\begin{split}
&\tr_{\reT} \bigg[ \cM_{\cl  \to \cl  \reT}^{t,k} (\rho_\cl ) \proj{k+l}_\reT\bigg] =\\
& \tr_{\reT} \bigg[ \cM_{\cl  \to \cl  \reT}^{t,k'} (\rho_\cl ) \proj{k'\!+l}_\reT\!\bigg]
\end{split}
\end{align}
for all $t\geq 0$, $\rho_\cl \in\mathcal{S}(\mathcal{H_\cl })$ and $l\in\zz$, $k,k'\in (0,1,\ldots,N_T)$ such that $k+l, k'\!+l\in (0,1,\ldots, N_T)$. 
Physically, this condition means that the dynamics of the \clock{} is invariant under translation of the input and output states of the register by the same amount. This is what one expects from a ticking clock; e.g.\ the probability of a ticking clock ticking 2 hours in the future according to coordinate time (and the state of the clockwork at this time), given that the clock's register was initiated to 3pm, should be the same as if it were initiated to 6pm.\Mspace

The instances of \cref{eq:clock invariance} for which $l$ is negative correspond to the state of the \clock{} when the register is found to have ``un-ticked'', i.e.\ that one finds the register to be in a state corresponding to an \emph{earlier} time than it was initiated to, while coordinate time has increased.\Mspace 

\emph{\textbf{\textup{2)}} The self-timing condition}: For all $t\geq 0$, the 
 ticking clock channel is Markovian:
\begin{align}
\cM^{t_1+t_2}_{\cl\reT  \to \cl\reT }&=\cM^{t_2}_{\cl\reT  \to\cl \reT} \circ \cM^{t_1}_{\cl\reT  \to \cl \reT}\label{eq_extension CR to CR}
\end{align}
for all $t_1,t_2\geq 0$.
This condition implies that no temporal information can come from systems alien to the \clock{} and register. See \cref{sec:Self-timing} for more details.\Mspace

\emph{\textbf{\textup{3)}} The zeroth order condition}:
\begin{subequations}
\begin{align}
&\cM_{\cl\reT  \to \cl\reT }^{0} =\idch_{\cl\reT}\label{eq:id channel clockwork a}\\
&\lim_{t\rightarrow 0^+} \| \cM_{\cl\reT  \to \cl\reT }^{t} -  \idch_{\cl\reT}\| = 0.\, \footnotemark \label{eq:id channel clockwork b}
\end{align}
\end{subequations}
\addtocounter{footnote}{-1}\footnotetext{This equation is known as {uniform} continuity in the semigroup literature~\cite{Pazy1983}.}
This condition simply states that if no time has passed, then no change is permitted in the ticking clock.
\Mspace

The next condition concerns the probabilities of ticks. We denote the probability that {the register is in the state corresponding to the $l^\text{th}$ tick}, at coordinate time $t$, given the $\proj{k}_\reT$ register state as input, by 
\begin{align}\label{eq:tilde p k l delta}
\tilde p_l^{(k)}(t) &:= \tr \bigg[ \cM_{\cl  \to \cl  \reT}^{t,k} (\rho_\cl )\proj{l}_\reT \bigg].
\end{align}

\emph{\textbf{\textup{4)}} The leading order condition}:
\begin{align}
\begin{split}\label{eq:leaidng order condition}
\lim_{t\to 0^+} \frac{\sum_{\underset{l\notin  \{k,  f(k)\}}{l=0\qquad}}^{N_T} \tilde p_{l}^{(k)}(t)}{\tilde p_{f(k)}^{(k)}(t)}=0,
\end{split}
\end{align}
for all $\rho_\cl\in\mathcal{S}\left(\mathcal{H}_\cl\right)$ where $f(k)=k+1 \textup{ mod. }N_T+1$.  
This condition imposes the constraint that the clock cannot ``skip a tick''. More precisely, between a coordinate time at which $k$ ticks have occurred, and a later coordinate time at which $l>k$ ticks have occurred, the probability that ticks $k+1, k+2,\ldots, l-1$ have occurred is one.\Mspace

Conditions 1) to 4) provide the necessary ingredients to define a ticking clock with a periodic register, but before doing so, it is advantageous to consider a distinct scenario which we call \emph{cut-off register}.\Mspace

In this scenario, the register will stop changing when it is full. For {$N_T=60$, an analogous wall clock would be one which you start at 8:00 and it stops ticking one hour later at 9:00.} Both cut-off and periodic register models have clear advantages and disadvantages: While the periodic ticking clock will never stop ticking, one can only determine the time up to multiples of its period (although this can be circumvented by counting the ticks in real-time); but while this issue does not arise in the cut-off case, it is only useful for keeping track of time until it runs out of memory. Both types of register exhibit some common characteristics, see \cref{Fig:memory}~b).\Mspace

While conditions similar to 2) and 4) can be defined for the cut-off register, due to the asymmetry in its boundary conditions, it is complicated to do so. A more direct and intuitive requirement is the following:\Mspace

\emph{\textbf{\textup{5)}} The cut-off register condition}: for every ticking clock channel with a cut-off register $ \cM_{\cl \reT \to \cl \reT}^t $, there exits a ticking clock channel with a periodic register denoted $\tilde \cM_{\cl \reT \to \cl \reT}^t $ and satisfying conditions 1) to 4),  such that in the $ t\to 0^+$ limit:
\begin{subequations}\label{eq:condition 5}
\begin{align}\label{eq:condition 5 a}
\cM_{\cl \to \cl \reT}^{ t,k} (\rho_\cl) = 
\tilde \cM_{\cl  \to \cl \reT}^{ t,k} (\rho_\cl)+ \lo(t)
\end{align}
for $k=0,1,\ldots, N_T-1$ and 
\begin{align} \label{eq:condition 5 b}
\tr_\cl\!\left[\cM_{\cl  \to \cl \reT}^{ t, N_T} (\rho_\cl)\right] = \proj{N_T}_\reT \!+ \lo( t),
\end{align}
\end{subequations}
for all $\rho_\cl \in\mathcal{S}(\mathcal{H_\cl })$ and where $\lo(\cdot)$ is little-o notation. This requirement captures some of the behaviour of the ticking clock with a periodic register, while enforcing the condition that given the last register state $\proj{N_T}$ as input, it can no longer invoke a change in the register | i.e. it ``stops ticking".\Mspace

After these general remarks, we are now ready to state the technical definition of a ticking clock:
\begin{definition}[Ticking clock]\label{def_quantumclock}
A \emph{ticking clock} is a pair $(\rho_{\cl\reT} ^0, (\smash{\cM^{t}_{\cl \reT \to \cl  \reT}})_{t\geq 0})$, with $\rho_{\cl\reT} ^0\in \mathcal{S}(\mathcal{H}_\cl\otimes\mathcal{H}_\reT)$ the state of the \clock{} and register at coordinate time $t=0$, and where the interaction between the \clock{} and register, governed by the channel $\cM^{t}_{\cl \reT \to \cl  \reT}$, satisfies conditions 1) to 4) 
 in the case of a periodic register, and conditions 2) and 5) in the case of a cut-off register.
\end{definition}

Observe that for a ticking clock with a cut-off register, the state of the \clock{} given the final state of the register $\proj{N_T}$ as input, is of no relevance, since it does not affect the state of the register (and hence the tick statistics) anymore. We emphasize this point with the following definition. {
\begin{definition}[\Clock{} equivalence]\label{def:equivalence}
	Two ticking clocks with a cut-off register are said to be \emph{\clock{} equivalent} if the two following conditions are both satisfied:
	
	\noindent {\bf a)} Their underlying ticking clock channels with a periodic register in \cref{eq:condition 5 a}, namely $\tilde \cM_{\cl \reT \to \cl \reT}^t $, are the same in both cases.
	
	\noindent {\bf b)} The states of their \clock{} when inputting the register state $\proj{N_T}$, namely 
	\begin{align}
	\tr_\reT\!\left[\cM_{\cl \to \cl \reT}^{t, N_T} (\rho_\cl)\right],
	\end{align}
	differ for some $t\geq 0$ and $\rho_\cl\in\mathcal{S}\left(\mathcal{H}_\cl\right)$.
\end{definition}
}

At present the registers are a collection of $N_T+1$ orthonormal states without spatial structure. One may furthermore demand that local structure is given to the register states via some distance measure $\text{Dis}\left(\ket{l}_\reT,\ket{m}_\reT\right)\geq 0$ that satisfies $\text{Dis}\left(\ket{l}_\reT,\ket{m}_\reT\right)<\infty$ for all $l,m=0,1,\ldots,N_T$. This requirement is physically motivated by noting that when it is imposed, and the register has a local structure, condition 4) (\cref{eq:leaidng order condition}) implies that the \clock{} does not have to ``travel'' an \emph{infinite} distance in finite time to write the next tick to the register | which would {be} unphysical. Furthermore, one can minimise the ``speed of sound'' in the register by arranging the local sites on the register so that $\text{Dis}\left(\ket{l}_\reT,\ket{m}_\reT\right)=g(|l-m|)$ for some monotonically increasing function $g$ in the case of a cut-off register.\footnote{In the periodic register case, one would use a period version of this.}  The simplest example of such a set-up is when the register is embedded into $\mathcal{H}^{\otimes {N_T+1}}$ where $\mathcal{H}$ is the space of a qubit spanned by $\ket{0}$, $\ket{1}$ and we identify $\ket{n}_\reT=\ket{1}^{\otimes n} \otimes \ket{0}^{\otimes N_T+1-n}$, and $\text{Dis}\left(\ket{l}_\reT,\ket{m}_\reT\right)=|l-m|$; $l,n=0,1,\ldots, N_T$. In this case, at the instance when the  $(k\!+\!1)^\text{th}$ tick occurs, the \clock{} ``flips'' the qubit ``next to'' the $k^\text{th}$ qubit.\Mspace

 Temporal information emitted from a clock in the standard treatment is classical in nature. Therefore, one should foremost consider \emph{classical registers} $\reT$. These are registers {for which the action of the \clock{} on them does not} produce coherence in the fixed basis $\{\ket{0}_\reT, \ket{1}_\reT,\ldots,\ket{N_T}_\reT\}$, when the register input is diagonal in this basis; namely dynamics of the form:
\begin{align}\label{eq:classical reg def}
 \cM_{\cl  \to \cl  \reT}^{t,k} (\rho_{\cl }^0) =\sum_{n=0}^{N_T}\rho_{\cl}^{(n;k)}(t)\otimes\ketbra{n}{n}_\reT,
\end{align}
for all $ k=0,1,\ldots,N_T;\, t\geq 0;\, \rho_\cl\in\mathcal{S}\left(\mathcal{H}_\cl\right)$, where $\rho_{\cl}^{(n;k)}(t)$ are arbitrary subnormalised states on the \clock. Considering dynamics of this form is also equivalent to assuming that the register is measured in the basis $\{\ket{0}_\reT, \ket{1}_\reT,\ldots,\ket{N_T}_\reT\}$ associated with ticks, since in this case any coherence will be destroyed and does not affect the measurement. Therefore, requesting dynamics of the form \cref{eq:classical reg def} is not to be confused with demanding that the register is itself classical, but moreover can be interpreted as requiring that we are only allowed to extract classical information from it. We will therefore take \cref{eq:classical reg def} as the defining property of a classical register: 
\begin{definition}[Classical register]\label{def:classical register}
A ticking clock $(\rho_{\cl\reT}^0, (\smash{\cM^{t}_{\cl \reT \to \cl  \reT}})_{t\geq 0})$ has a \emph{classical register} if the channels $\cM^{t,k}_{\cl \to \cl  \reT}$ are of the form \cref{eq:classical reg def} for all $t\geq 0$, $k=0,1,\ldots,N_T$.
\end{definition}

Classical registers have the advantageous property that they can be ``continuously observed'' without changing the properties of the ticking clock | analogously to how one can continuously look at a wall clock or listen for its ticks without disturbing the dynamics of its clockwork (and hence accuracy). This may come as a surprise for two reasons: For one, clearly the state of the ticking clock in \cref{eq:classical reg def}, before and after measuring the register in the basis $\{\ket{n}_\reT\}$ is different. Secondly, the Zeno effect~\cite{zeno1,zeno2,zeno3} dictates that if a quantum system is continuously measured, then it will stop evolving altogether | which is clearly at-odds with the desired properties of a ticking clock.\Mspace

The solution to the 1st apparent problem is to recall that 
the state of the register is a probabilistic mixture and thus the change in the state due to the register's measurement is due to a change in our knowledge about which state the register is in. This is analogous to the description of any purely classical ticking clock which is not perfectly accurate:  while in every run of the ticking clock in which it is continuously observed, the state of the register will always be known exactly; in order to calculate the statistics associated with its accuracy, one will need the ensemble of all possible ticking clock runs weighted by the probability that each trajectory occurs. The 2nd apparent problem is resolved by showing that the Zeno effect does not apply to continuous measurements of the register when it is of the form \cref{eq:classical reg def}.\Mspace

The following proposition formalises the previous two remarks by showing that one can continuously measure the register without affecting the statistics associated with the probabilistic distribution of ticks. Before stating it, we need to introduce some notation and definitions:

Let $\mathcal{P}_l(\cdot):=\proj{l}_\reT (\cdot) \proj{l}_\reT /\tr[(\cdot)\proj{l}_\reT] : \bounded(\mathcal{H}_\cl\otimes\mathcal{H}_\reT)\to \bounded(\mathcal{H}_\cl\otimes \mathcal{H}_\reT)$ denote the channel which takes any ticking clock state and outputs the state of a ticking clock after measuring the register in the register basis $\{\ket{0}_\reT, \ket{1}_\reT,\ldots,\ket{N_T}_\reT\}$ and finding the register to be in the state $\proj{l}_\reT$. 

\begin{definition}[Measured channels]\label{def:continous measured channel}
	Given a ticking clock $(\rho_{\cl\reT}^0, (\cM^{t}_{\cl\reT \to \cl \reT})_{t\geq 0})$, we call the following channel $\bounded(\mathcal{H}_\cl\otimes\mathcal{H}_\reT)\to\bounded(\mathcal{H}_\cl\otimes\mathcal{H}_\reT)$ a \emph{measured channel:}\footnote{We use the notation 
	$\circsum{1}{N} f_n(\cdot):=f_N \circ f_{N-1}\circ\ldots\circ f_1(\cdot)$.}
	\begin{align}
	\mathcal{CM}^t_{\cl\reT\to\cl\reT}[\vec s_N](\cdot):=
	\circsum{1}{N} \left(\mathcal{P}_{l_n}\circ \cM^{t_n}_{\cl\reT \to \cl \reT}\right)(\cdot),\label{defeq:continous measured channel}
	\end{align}
	where $\vec s_N:= (l_n,t_n)_{n=1}^N$ is the sequence of measurement outcome indices $l_n=0,1,\ldots,N_T$ and times $t_n\geq 0$; s.t. $\sum_{n=1}^N t_n=t$. In the case that $\cM^{t}_{\cl\reT \to \cl \reT}$ has a classical register we call the channel a \emph{classical register measured channel}.
\end{definition}

The channel \cref{defeq:continous measured channel} corresponds to the state of the ticking clock at coordinate time $t$ when the free evolution of the ticking clock was interrupted at times $t_n$ by register measurements with outcomes $\proj{l_n}_\reT$. 
Let $\mathrm{Prob}[{\vec s}_N]$ be the probability that the sequence of outcomes with indices $l_1,l_2,\ldots,l_N$ at times $t_1,t_2,\ldots,t_N$ occurs. We denote the set of all sequences of outcomes at times $(t_1,t_2,\ldots,t_N)=:\vec t$ by
\begin{align}
 \mathbf{\mathtt{S}}_N(\vec t)\!:= \big\{ (l_n,t_n)_{n=1}^N : l_n\in\{0,1,\ldots,N_T\}
 \big\}\!.
\end{align}
\begin{restatable}[Measured register equivalence]{prop}{continuousMeasurements}\label{prop:continous measurements}
For all coordinate times $t_n\geq 0$ s.t. $\sum_{n=1}^N t_n=t$ and for all $N\in\nnp$, the dynamics of any ticking clock with a classical register is equal to that of the ensemble of classical register measured channels, where the ensemble is weighted by the probability of the classical register measured channel occurring:
\begin{align}\label{eq:equivalnce under meres register}
&\cM^t_{\cl\reT\to\cl\reT}(\rho_\cl\otimes\proj{k}_\reT)\\
&=\sum_{{\vec s}_N\in  \mathbf{\mathtt{S}}_N(\vec t)}\mathrm{Prob}[ \vec s_N]\, \mathcal{CM}^t_{\cl\reT\to\cl\reT}[ \vec s_N](\rho_\cl\otimes\proj{k}_\reT)\nonumber
\end{align}
for all $(t_n)_{n=1}^N$, $N\in\nnp$, $t\geq 0$, $k=0,1,\ldots, N_T$, and $\rho_\cl\in\mathcal{S}(\mathcal{H}_\cl)$.
\end{restatable}
See \cref{sec:measurement equiavalence} for proof.
A direct consequence of \cref{prop:continous measurements} is that if we choose $t_n= t/N$ followed by taking the limit $N\rightarrow \infty$ for fixed $t$ on the r.h.s. of \cref{eq:equivalnce under meres register}, we are in the regime of continuous measurements proposed in the Zeno effect. However, in this continuous measurement case, \cref{prop:continous measurements} certifies that the ticking clock channel $\cM^t_{\cl\reT\to\cl\reT}(\rho_\cl\otimes\proj{k}_\reT)$ still adequately describes the statistics. Indeed, if the register starts in the state $\proj{k}_\reT$, the probability $\mathrm{Prob}[ \vec s_N]$ specialised to the case of finding the register in the state $\proj{k}_\reT$ for all time $t\in [0,\tau]$, 
for some $\tau>0$ would have to be one if Zeno's mechanism were to hold. We will later see that it is however only true for some irrelevant trivial clocks.\Mspace

For later purposes, it is useful to introduce a notion of a ``classical ticking clock". This notion of classicality is effectively the same as the one introduced in~\cite{woods2018quantum} but stated for the ticking clock introduced in this \doc{} (\cref{def_quantumclock}).

\begin{definition}[Classical ticking clock]\label{def:classical tickinc clock} We call a ticking clock $(\rho_{\cl\reT}^0,(\cM^t_{\cl\reT\to\cl\reT})_{t\geq 0})$ \textup{classical}, if there exists a basis $\{\ket{l}_l\}_{l}$ spanning the \clock{} Hilbert space $\mathcal{H}_\cl$, for which the \clock{} remains incoherent in this basis at all coordinate times:
\begin{align}
	\tr_\reT\left[\cM^{t}_{\cl\reT \to \cl\reT}(\rho_{\cl\reT}^0)\right]=\sum_{l} p_l(t)\proj{l}_\cl,\quad \forall\, t\geq 0.
\end{align}
\end{definition}
Likewise, we call a ticking clock a \emph{quantum ticking clock} if it does not satisfy the classical ticking clock criterion. Thus unless otherwise specified, a ticking clock may be quantum or classical.

\section{Autonomous dynamics}\label{sec:Explicit Ticking Clock Representation}

In this section we formulate a representation of the ticking clock channel $\cM^{t}_{\cl\reT \to \cl \reT}$ which holds if and only if the ticking clock satisfies the axiomatic \cref{def_quantumclock}, up to some stated equivalence. An alternative | more technical in nature | representation is left to \app{} \cref{Sec:Implicit Ticking Clock Representation}. The highly inquisitive reader may want to detour to \cref{Sec:Implicit Ticking Clock Representation} before continuing here, while the more cursory reader may avoid \cref{Sec:Implicit Ticking Clock Representation}.


\begin{restatable}[Explicit ticking clock representation]{prop}{ExplicitTickingClockForm}\label{prop:autonmous Lindbald form}
	The pair $(\rho_\cl^0, (\cM^{t}_{\cl \reT \to \cl \reT})_{t\geq 0})$ form a ticking clock (\cref{def_quantumclock}) with a classical tick register (\cref{def:classical register}), up to \clock{}  equivalence (\cref{def:equivalence}), if and only if there exists a Hermitian operator $H$ as well as two finite sequences of operators $(L_j)_{j=1}^\m$ and $(J_j)_{j=1}^\m$ on $\bounded(\mathcal{H}_\cl)$; which are all $t$ independent, such that for all $t\geq 0$ and $N_T\in \nnp$,
\begin{subequations}
	\begin{align}
	&\cM_{\cl\reT \to \cl\reT}^{t} (\cdot)=\; \me^{t \mathcal{L}_{\cl\reT}}(\cdot),\label{eq:clock dynamics landblad}\\
	&\mathcal{L}_{\cl\reT} \left( \cdot \right) = -\mi[\tilde{H},(\cdot)] + \sum_{j=1}^\m  \tilde{L}_{j} (\cdot) \tilde{L}^\dagger_{j} - \frac{1}{2} \big\{ \tilde{L}^\dagger_{j} \tilde{L}_{j}, (\cdot) \big\}\nonumber\\
	&\,\,\,\,\qquad\qquad + \sum_{j=1}^\m \tilde{J}_j^{(l)} (\cdot) \tilde{J}_j^{(l)\dagger} - \frac{1}{2} \big\{ \tilde{J}_j^{(l)\dagger} \tilde{J}_j^{(l)}, (\cdot) \big\},\label{eq:clock dynamics landblad eq}
	\end{align}
	where the operators are $\tilde{H} = H \otimes \id_\reT$, $\tilde{L}_{j} = L_j \otimes \id_\reT$, $\tilde{J}_j^{(l)} = J_j \otimes O_\reT^{(l)}$, with
\begin{align}
\begin{split}
O_\reT^{(l)}:=& \ketbra{1}{0}_\reT+\ketbra{2}{1}_\reT+\ketbra{3}{2}_\reT+\ldots+\\
& \ketbra{N_T}{N_T-1}_\reT+l\ketbra{0}{N_T}.\label{eq:local O operator}
\end{split}
\end{align}
\end{subequations}
In the cut-off register case $l=0$, while $l=1$ for the periodic register case.
\end{restatable}
See \ref{sec: proof of autonmous Lindbald form prop} for a proof. It follows straightforwardly from a more technical representation discussed in \app\Mspace

Observe that the Lindbladian in \cref{eq:clock dynamics landblad eq} only requires local coupling between the orthogonal register states, according to the distance measure $\text{Dis}(\cdot,\cdot)$ introduced in \cref{sec:new quantum clock def} after \cref{def_quantumclock}. The dynamics of the register also clearly satisfies the finite running memory condition in \cref{sec:Finite memory}.\Mspace

While the ticking clock model in~\cite{RaLiRe15}, in the case of continuous coordinate time $t$, does have a dynamical semigroup representation from the \clock{} and \emph{individual} tick registers to itself, $\lineaR(\mathcal{H}_\cl\otimes\mathcal{H}_{\re_j})\to\lineaR(\mathcal{H}_\cl\otimes\mathcal{H}_{\re_j})$ for all $j\in\nn$, (see~\cite{woods2018quantum}) a dynamical semigroup representation on the \clock{} and the \emph{total} tick register has not been shown to exist. As explained previously in \cref{sec:the ATG model}, since the dynamics of the register is dependent on the details of the gear system in their model, such a formulation would inevitably need to include a description of the gear system used, which, for any realistic gear system, would arguably lead to unaccounted for sources of error.\Mspace

The following proposition shows that an effective Markovian dynamical semigroup for the \clock{} can always be found. It therefore justifies associating the \clock{} with the sole source of timing.

\begin{restatable}[\Clock{} representation]{prop}{clockchannel}\label{prop:clockwork channel}
Consider a ticking clock 
with a classical periodic register (\cref{def:classical register,def_quantumclock}) written in the representation of \cref{prop:autonmous Lindbald form}. Its \clock{} channel, defined via
\begin{align}
\cM_{\cl\to\cl}^{t}(\cdot):=\tr_\reT[\cM_{\cl\reT\to\cl\reT}^{t}((\cdot)\otimes\proj{k}_\reT)]\label{eq:clockwork channel def}
\end{align}
is $k$-independent for all $t\geq 0$ and of the form 
\begin{align}
\cM_{\cl\to\cl}^{t}(\cdot)= \me^{t \mathcal{L}_\cl}(\cdot), \label{eq:clockwork generator}
\end{align}
with $\mathcal{L}_\cl$ equal to the r.h.s. of \cref{eq:clock dynamics landblad eq} under the replacements $\tilde H\mapsto H$, $\tilde L_j\mapsto L_j$ and $ \tilde J_j^{(l)}\mapsto J_j$. What is more, for every ticking clock with a classical cut-off register written in the representation of \cref{prop:autonmous Lindbald form}, there exists a ticking clock which is \clock{} equivalent (\cref{def:equivalence}), such that its \clock{} is $k$-independent and given by \cref{eq:clockwork generator}.
\end{restatable}

The proof is constructive and can be found in \cref{sec: proof of clockwork channel prop}. In the case of the cut-off register, the representation used in the proof for which \cref{eq:clockwork generator} holds, has a ticking clock channel whose \clock{} still produces ticks when the input register state is $\proj{N_T}_\reT$, but does not write them to the register. This can be contrasted with the \clock{} equivalent ticking clock with cut-off register representation in \cref{prop:autonmous Lindbald form}. In this case, the \clock{} stops producing ticks when the input register state is $\proj{N_T}_\reT$.\Mspace

{Dynamical semigroups, such as the clockwork channel representation in \cref{eq:clockwork generator}, do not have an inverse channel; and hence ticking clocks emit temporal information in an irreversible fashion.}\Mspace

Dynamical semigroups of the form \cref{eq:clock dynamics landblad} have been shown to have a microscopic description in which the system (which in the present case would constitute the \clock{} and total register) interacts with an infinite dimensional environment on $\mathcal{H}_\textup{E}$ via a time independent Hamiltonian under the appropriate limits. In particular, the Hamiltonian $H_\textup{tot}$ which leads to dynamics \cref{eq:clock dynamics landblad,eq:clock dynamics landblad eq} is of the form
\begin{align}
H_\textup{tot}=H\otimes\id_{\reT\textup{E}} -\tilde{H}_{\cl\reT}\otimes\id_\textup{E}+ \id_{\cl\reT}\otimes H_\textup{E}+ V,\label{eq:Hamilotnian formulation}
\end{align}
where $\tilde{H}_{\cl\reT}\in\bounded(\mathcal{H}_\cl\otimes\mathcal{H}_\reT)$ is tuned to counteract a shift in energy on the \clock{} and register due to interactions with the environment with local Hamiltonian $H_\textup{E}$, while $V$ mediates the interaction between the register, \clock{}, and environment. It takes on the form $V=\sum_n \id_\reT \otimes A^{(L)}_n\otimes B_n^{(L)}+ O_\reT^{(l)}\otimes A_n^{(J)} \otimes B_n^{(J)}$ with $O_\reT^{(l)}$ given by \cref{eq:local O operator}. The $(A^{(L)}_n)_n$ and $(A^{(J)}_n)_n$ terms give rise to the operators $(L_j)_j$ and $(J_j)_j$ respectively while terms $(B_n^{(L)})_n, (B_n^{(J)})_n$ are suitably chosen local terms acting on the environment. 
There are two known types of limiting procedures which one can apply to \cref{eq:Hamilotnian formulation} to achieve dynamics of the form \cref{eq:clock dynamics landblad,eq:clock dynamics landblad eq}. One in which the time-scales of the environment are much shorter than those of the system | the so-called ``weak coupling limit'' | and the other where the time scales are reversed | called the ``singular coupling limit''. See~\cite{Breuer2007} for physical insight into these limits and~\cite{Palmer77} for how to re-scale the interactions to interchange between them. In~\cite{GORINI1978149} it is proven that \emph{all} Lindblad operators in \cref{eq:clock dynamics landblad eq} are achievable via appropriate choice of the terms in \cref{eq:Hamilotnian formulation}, in the singular coupling limit. Finally, observe that the interaction term $V$ in \cref{eq:Hamilotnian formulation} only requires local coupling to the register.


\section{Ticking Clock Examples}\label{sec:Ticking Clock Examples}
In this section we will see how clocks from the literature are either a special case of the ticking clocks introduced here (\cref{def_quantumclock}), or that they can be easily adapted to be of this form. In all three examples, we provide the choice of $H$, $(L_j)_{j=1}^\m$ and $(J_j)_{j=1}^\m$  from \cref{prop:autonmous Lindbald form} for the cut-off register case. While the examples are also valid in the periodic case, they are distinct to those in the literature for this choice. The accuracy of these example clocks will be discussed in \cref{sec:accuracy}.

\subsection{Thermodynamic Ticking Clock}\label{sec:example thermo}
The clock in~\cite{thermoClockErker} is a ticking clock in which population is driven up a $d$ dimensional ``ladder'' via the interaction 
with two thermal qubits maintained at thermal equilibrium at distinct hot and cold temperatures through their coupling to local thermal baths. A tick occurs when the population of the ladder reaches the top via spontaneous emission back to the ground state of the ladder. The tick is recorded in the tick register by a photo-detector. This register is of the same form as the one introduced here: it is classical in nature and is only written to when a tick occurs. It thus satisfies both the self-timing and finite running memory conditions of \cref{sec:basic principles}. The thermodynamical clock in~\cite{thermoClockErker}, is specified by choosing $N_L=4$ with
\begin{subequations}
\begin{alignat}{2}
L_1&=\sqrt{\gamma_h} \sigma_h,\quad &&L_2= \sqrt{\gamma_h \me^{-\beta_h E_h}} \sigma_h^\dag,\\
L_3&=\sqrt{\gamma_c} \sigma_c,\quad &&L_4= \sqrt{\gamma_c \me^{-\beta_c E_c}} \sigma_c^\dag,\\
J_1&= \sqrt{\Gamma} \ketbra{0}{d-1}_w,\quad &&\,J_2=J_3=J_4=0,
\end{alignat}
\end{subequations}
where $\sigma_h, \sigma_c$ are lowering operators for the hot and cold qubits respectively; $\ket{0}_w$, $\ket{d-1}_w$ are the ground and top states of the ladder, and the other coefficients are positive and defined in~\cite{thermoClockErker}. The Hamiltonian takes on the form $H= H_0+ H_\textup{int}$ where $H_0$ is the local Hamiltonian for the qubits and ladder while $H_\textup{int}$ is the three-body interaction between them. Indeed, a simple calculation of the state of the \clock{} $\rho_\cl^{(0)}(t)$ given that the 1st tick has not occurred in the time interval $[0,t]$, using \cref{eq:eqfor COM} and the above parameters yields $\rho^{(0)}_{\cl}(t) = \me^{t \mathcal{L}_{\cl}^{(0)}} (\rho_\cl)$ with

\begin{align}\label{eq:lidblad for thermo}
\begin{split}
\mathcal{L}_{\cl}^{(0)}(\cdot) =&\, \mi \left( \hat H_\textup{eff}(\cdot)- (\cdot) \hat H_\textup{eff}^\dag  \right)\\
 &- \sum_{j=1}^4 \frac{1}{2} \left\{ L_j^\dag L_j, (\cdot)  \right\} + L_j (\cdot)L_j^\dag,
\end{split}
\end{align}

where $\hat H_\textup{eff}= H-\mi\, \Gamma \proj{d-1}_w$. This is in exact correspondence with eq. B4 on page 8 of~\cite{thermoClockErker}. Since the \clock{} is reset to its initial state after each tick, \cref{eq:lidblad for thermo} fully determines the statistics of all ticks as discussed in~\cite{thermoClockErker}.\Mspace

\subsection{Quasi-ideal Ticking Clock}\label{sec:example QIC} 

In~\cite{woods2018quantum} a ticking clock based upon the results from~\cite{WSO16} was introduced in the context of the model from~\cite{RaLiRe15}. Here we consider the same \clock{} but when its coupling to the register results from the axioms introduced in \cref{sec:new quantum clock def}, rather then the model~\cite{RaLiRe15}. It therefore satisfies \cref{def_quantumclock} of a ticking clock. From~\cite{woods2018quantum} we have that $N_L=d$, the dimension of the \clock{} and for all $j=1,2,\ldots, d$:
\begin{align}
L_j=0, \quad J_j=\sqrt{2 V_j} \ketbra{\psi_\cl}{t_j},\label{eq:lindbalds for QI clock} 
\end{align}
where $(\ket{t_j})_{j=1}^d$ is an orthonormal basis for $\mathcal{H}_\cl$. The state $\rho_\cl^0=\ketbra{\psi_\cl}{\psi_\cl}$ is both the initial state of the \clock{} and  the state which it is reset to after each tick. It is called the Quasi-ideal clock state and follows a complex Gaussian distribution in the $(\ket{t_j})_{j=1}^d$ basis. The coefficients $(V_j>0)_{j=1}^d$ follow a peaked distribution; see~\cite{woods2018quantum} for details. The Hamiltonian $H$ is a ladder Hamiltonian with equally spaced energy gaps and diagonal in the Fourier transform basis generated from $(\ket{t_j})_{j=1}^d$.\Mspace

\begin{figure}\includegraphics[scale=0.55]{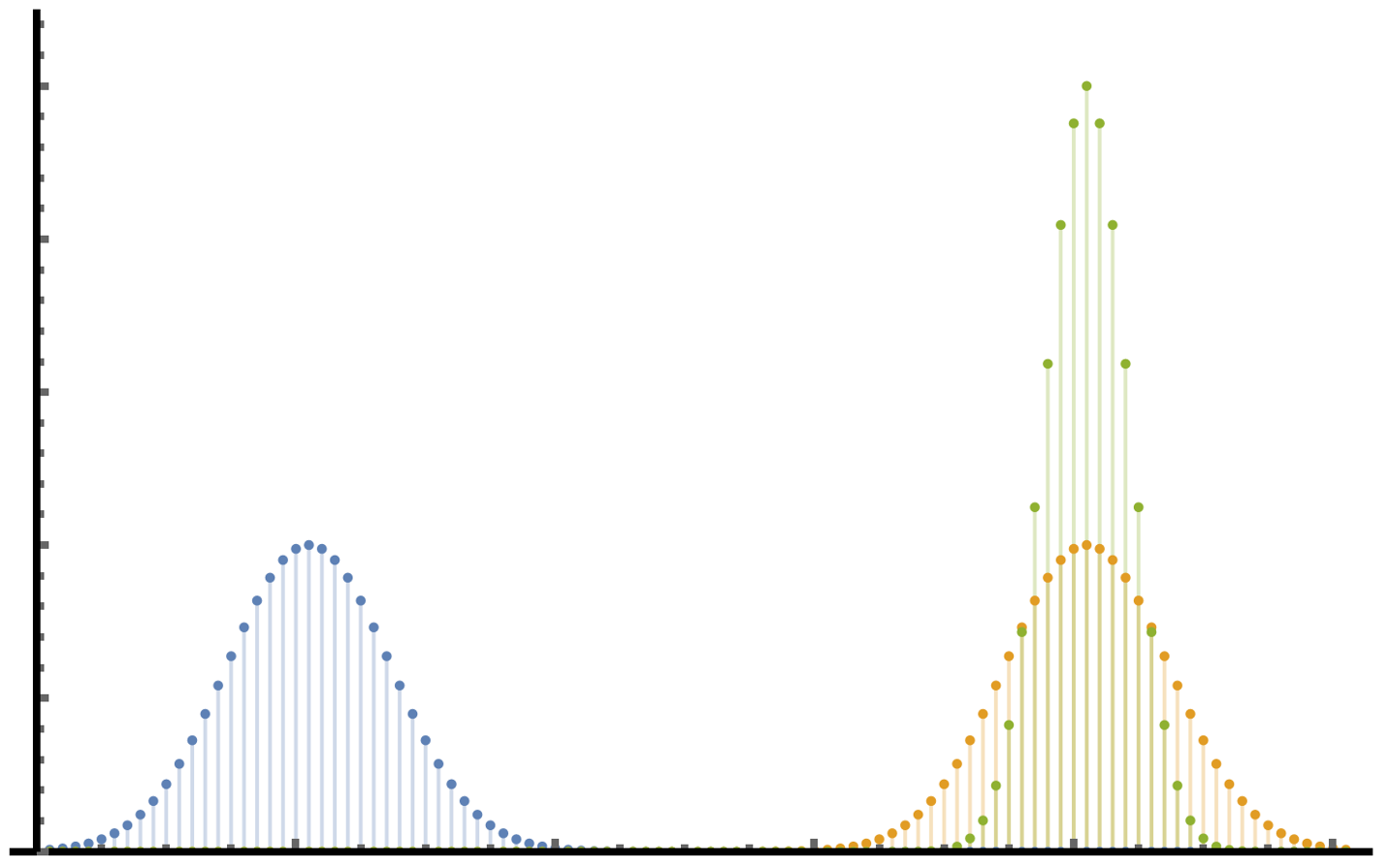}
\caption{{Qualitative plots for the Quasi-ideal ticking clock. Blue: magnitude of the amplitudes of the \clock{} in the basis $(\ket{t_j})_{j=1}^d$ at time $t=0$. Orange: magnitude of amplitudes of the \clock{} in the basis $(\ket{t_j})_{j=1}^d$ around the time when the 1st tick is likely to occur. Green: coefficients $(V_j)_{j=1}^d$. Observe the the blue and green curves have very small overlap while the orange and green have a large overlap.
	}}\label{fig:quasi ideal clock}
\end{figure}

The free dynamics of the \clock{} according to $H$ allows the complex Gaussian amplitude distribution to ``move coherently'' in the $(\ket{t_j})_{j=1}^d$ basis until the peak of the distribution overlaps with the peak of the distribution $(V_j>0)_{j=1}^d$, at which point a tick occurs and the \clock{} is reset and starts again (see \cref{fig:quasi ideal clock}). Note that the statistics of the 1st tick are invariant under the exchange of the Lindblad operators in \cref{eq:lindbalds for QI clock} with the simpler form
$L_j=J_j=0$ for $j=1,2,\ldots, d-1$ and $L_d=0$, $J_d=\sum_{j=1}^d \sqrt{2 V_j} \ketbra{j}{t_j}$. Here $\{\ket{j}\in\mathcal{H}_\cl\}_j$ is any orthonormal basis since in this case the \clock{} does not need to be reset after it ticks | as can be seen formally from \cref{eq:eqfor COM}. 

\subsection{Ladder Ticking Clock}\label{sec:example Ladder} 
This clock is a classical ticking clock (\cref{def:classical tickinc clock}) which was defined in~\cite{ATGRandomWalk} and proven to be the most accurate classical continuous coordinate time clock in~\cite{woods2018quantum} in the context of the model~\cite{RaLiRe15}. As with the example of \cref{sec:example QIC}, it can easily be adapted to the ticking clock model {in} \cref{def_quantumclock}. We find $N_L=d$, with
\begin{subequations}
\begin{alignat}{2}
L_j&= \ketbra{c_{j+1}}{c_j},\quad  &&J_j= 0,\\
L_d&= 0,   &&J_d= \ketbra{c_1}{c_d},
\end{alignat}
\end{subequations}
for some orthonormal basis $\{\ket{c_j}\}_{j=1}^d$. The {clock's} initial state is $\rho_\cl= \proj{c_1}$ and since it is a classical clock, the Hamiltonian term vanishes, $H=0$.\Mspace

{In an appropriate limit, this clock can also be approximated by the thermodynamic clock.  See~\cite{thermoClockErker} for details.}

\section{Measures of accuracy}\label{sec:accuracy}
We call an \emph{accuracy measure} of a ticking clock, any quantity which can be written solely as a function of measurement outcomes of the register on $\mathcal{H}_\reT$ at different coordinate times. It can only depend on the state of the \clock{} indirectly, via its coupling to the register, but not on the state of the \clock{} directly. This is important since with the accuracy measure one wants to capture how good the \clock{} is at emitting temporal information to the ``outside''. Good examples of accuracy measures are those which depend solely  on the tick delay functions. The $k^\text{th}$ tick \emph{delay function} is the probability density that $k-1$ ticks occur during the coordinate time interval $[0,t)$ while the $k^\text{th}$ tick occurs exactly at coordinate time $t$. It is worth noting that while in the theoretical description of the ticking clock model one has knowledge of the state of the register \emph{and} the corresponding value of coordinate time, in an actual physical implementation of any ticking clock, one can only infer coordinate time indirectly. As such, any measure of accuracy may be hard to calculate experimentally and may require many repeated experiments involving multiple copies of any given ticking clock in order to garner enough statistics. This constraint of the model, is a virtue not a weakness however, since by definition the only information we should have access to (i.e. which should be stored in a register), is the information about ticks | not coordinate time directly. The storage of coordinate time in the register was an inherent drawback of earlier ticking clock models.\Mspace

 We now consider the cut-off register model, since in this case, the $k^\text{th}$ tick delay function (for $k=1,2,3,\ldots, N_T$) takes on exactly the same expression in terms of how it is related to dynamics on the \clock{} as presented in~\cite{woods2018quantum} for the model~\cite{RaLiRe15} (see \cref{Sec:Implicit Ticking Clock Representation} for details). The clock accuracy of the $k^\text{th}$ tick $R_k$, introduced in~\cite{thermoClockErker,woods2018quantum}, is the square of the ratio between the mean and standard deviation of the $k^\text{th}$ tick delay function with respect to coordinate time $t$. As such, there is a one-to-one relation between the ticking clocks in~\cite{woods2018quantum} and the ones introduced in this \doc{} for $k=1,2,\ldots, N_T$. Subsequently, all the theorems about the accuracy of ticking clocks in~\cite{woods2018quantum} also apply to those introduced here with a cut-off register. For instance, the most accurate classical ticking clock (see \cref{sec:example Ladder}) satisfies $R_k=k d$ for $k=1,2,\ldots, N_T;\, d\in\nnp$; where $d$ is the Hilbert space dimension of the \clock{}. The thermodynamic clock in \cref{sec:example thermo} can also achieve a similar accuracy; see~\cite{thermoClockErker}. On the other hand, there exists a quantum ticking clock (see \cref{sec:example QIC}) whose accuracy is lower bounded by 
\begin{align}
R_k= k R_1, \quad R_1\geq  d^{2-\varepsilon} +\lo(d^{2-\varepsilon}),
\end{align}
for all $k=1,2,\ldots, N_T$ and for all fixed $\varepsilon>0$ in the large $d$ limit~\cite{woods2018quantum}. Recently it has been shown that this bound is essentially tight for ticking clocks which only tick once~\cite{Yuxiang2}. It remains an open question whether a quantum ticking clock which ticks more than once can have ticks which have a higher accuracy.\Mspace

The alternative ticks game measure of accuracy~\cite{RaLiRe15,ATGRandomWalk}, is applicable to the ticking clock models developed here, with the difference that the referee will need to play the game {on-the-fly}, rather than comparing register states at the end {of the protocol (as proposed in \cite{RaLiRe15}),} since the registers in the new model do not record the coordinate time corresponding to when the tick occurred. 
\Mspace

\section{Conclusions and Outlook}\label{sec:discussions}

\subsection{Conclusions}\label{Sec:conclusions}

We started by discussing the ticking clock model of~\cite{RaLiRe15} which was one of the first theoretical models of a ticking clock. This revolutionary work inspired follow-up papers yet also some legitimate concerns from the community\footnote{By ``community'' we refer to comments from anonymous referees and senior scientists during conference talks.} regarding the physicality of its foundations which we formalise and discuss.
We then introduce an axiomatic definition of a new ticking clock based on physical principles and derive explicit solutions to its equations of motion. The aforementioned drawbacks do not apply to the new ticking clock model introduced here. {In a nutshell, the main reasons are twofold: On the one hand the new formalism only requires finite running memory (see \cref{sec:Finite memory}), while one the other, the axioms imply the existence of a dynamical semi-group which describes the dynamics of the \emph{entire} register and clockwork for all coordinate time.} We furthermore show that the new equations of motion admit a fully autonomous realisation.\Mspace

With every ticking clock, one can associate a set of delay functions which determine the accuracy measures of the ticking clock. We show that there is a one-to-one correspondence between the set of delay functions produced by the new ticking clock model introduced here and that of~\cite{RaLiRe15}. Consequently, bounds on the accuracy of the clocks in~\cite{RaLiRe15}, derived in~\cite{woods2018quantum}, apply also to the new ticking clock model presented here. Therefore, the main conclusion of~\cite{woods2018quantum}, namely that quantum ticking clocks are more accurate than classical ones, applies also to the new ticking clock model presented here.\Mspace

Another ticking clock model, based on thermodynamic principles was introduced in~\cite{thermoClockErker}. This model has many positive points, such as being fully autonomous and physically well motivated. It does however have some drawbacks, such as not being derived from first principles and having a reported accuracy which is substantially lower than that of the quantum ticking clock in~\cite{woods2018quantum}. Consequently, the results of this \doc{} imply that both of the desirable properties of the ticking clock models~\cite{thermoClockErker,RaLiRe15} are achievable in one physically transparent model: full autonomy and high accuracy. This is considered the most import conclusion of this \doc. What's more we have seen that the ticking clock in~\cite{thermoClockErker} is in fact a special case of the ticking clock model introduced here.\Mspace

Here by \emph{autonomous} it is understood that for every ticking clock according to the new \cref{def_quantumclock}, there exists a large macroscopic environment, a time independent and local Hamiltonian over the \clock{}, register, and environment, which represents said ticking clock. Note that this environment need not necessarily be thermal. Other possibilities such as a vacuum state, may turn out to be necessary in some cases. Whether indeed a (or several) thermal baths at various temperature(s) are sufficient for the most accurate clocks is an open question.\Mspace

One would like to be able to continuously monitor the tick register of a ticking clock in real-time | analogously to how one can listen for the ticks of a wall clock in real-time. At first sight, this may seem impossible since the quantum Zeno effect dictates that continuous observation of a quantum system causes it to freeze its motion; which is clearly not a desirable property of a clock. We show that one can continuously observe the tick register without affecting its statistics and hence its accuracy.\Mspace 




\subsection{Outlook}

It is worth discussing some of the multiple future directions this work opens up. On the one hand there are questions such as what is the most accurate quantum clock, how much entropy is produced every time the ticking clock ticks, or how would one build such a ticking clock in practice. On the other hand, one can consider extensions to the formalism itself. One very practical such extension would be to include noise from the environment and the possibility of leveraging tools from quantum error correction to protect ticks against such adversarial noise. 
Such studies could also be applied to other types of extensions to the model. We provide 4 such examples:\Mspace

1) Clocks in a network: One can readily extend the model introduced in this \doc{} to take into account a source of external timing. One way to consider external timing in the case of a ticking clock was introduced in~\cite{Yuxiang} and formulated in terms of the continuous time limit of the ticking clock of~\cite{RaLiRe15}. This extension could also be formulated for the ticking clocks introduced in this \doc. We now discuss two other types of extensions which to date have not been considered in the literature thus far.\Mspace

2) Relativistic ticking clock model: The ticking clock models in the literature (including this one), are not relativistic. Making them so would be an interesting endeavour. {Finding a convincing axiomatic definition appears  feasible, but} since we do not have a fully credible theory of {quantum} gravity yet, {solutions might turn out to} be highly speculative. One approach is to attempt to make relativistic versions of the axioms for ticking clocks presented in \cref{sec:new quantum clock def}, by stating how the observable statistics and invariant quantities in these axioms transform relativistically. 
Another method would be to employ the same approach used in~\cite{Shishir} to construct a relativistic quantum stopwatch. In such a semi-classical approach, one would include in the ticking clock set-up an additional kinematic degree of freedom associated with the ticking clock's momentum and position. One would then expand to leading order in relativistic corrections the general relativistic equations for time dilation, following a similar procedure to as in~\cite{pikovski2015universal}.\Mspace

3) Relaxation of the axioms: One could consider variants of the ticking clock models presented here by changing or disregarding some of the conditions 1) to 5) in \cref{sec:new quantum clock def}. An obvious choice would be to consider ``Time variant asymmetric ticking clocks'', namely those which do not satisfy condition 1) [\cref{eq:clock invariance}]. An example of such a ticking clock channel with a classical register would be \cref{eq:clock dynamics landblad,eq:clock dynamics landblad eq}  in \cref{prop:autonmous Lindbald form} under the replacement $\tilde J_j^{(l)} \mapsto$ 
$\sum_{k=0}^{N_T} J_{j,k} \otimes \ketbra{k+1 \textup{ mod. } N_T+1}{k}_\reT \left( 1- \hat \delta_{k,N_T+l} \right),
$ 
where $J_{j,k}$ are arbitrary operators on the \clock. Note however that such generalisations clearly cannot improve the accuracy of the clock.\Mspace

4) Unitary ticking clock model: The model proposed in this \doc{} takes on the form of a one-parameter dynamical semigroup over the \clock{} and tick register. The clockwork provides the necessary timing while the register stores the tick information. The potential accuracy of the ticking clock depends on the properties of the \clock, such as its dimensionality or energy. We have discussed how this can be dilated via the aid of an infinite-dimensional environment to unitary dynamics with a time independent Hamiltonian, using standard limiting procedures.

One could consider a ticking clock whose dynamics are unitarily driven by a time independent Hamiltonian with a \emph{finite} environment instead. 
It would however have no classical counterpart (\cref{def:classical tickinc clock}) nor allow for a classical register (\cref{def:classical register}). The lack of a classical register would mean that it could suffer from the Zeno effect, and thus there would be a critical observation frequency which if surpassed, the observations would start to change the accuracy of the clock significantly and if frequent enough, might even stop the clock ticking altogether. Understanding and quantifying the properties of such models could be an interesting future line of research.\Mspace

\tocless\acknowledgments
We thank Tony Short for encouragement to investigate alternative clocks models to that of~\cite{RaLiRe15} and Marcus Huber, Robert Salzmann, Nilanjana Datta, Mark Mitchison, and Renato Renner for useful discussions. We thank {\'A}lvaro Alhambra for careful proofreading and insightful comments. This work was supported by the Swiss National Science Foundation (SNSF) via an AMBIZIONE Fellowship (PZ00P2\_179914) in addition to the National Centre of Competence in Research ``QSIT''.




\bibliographystyle{apsrev4-1}
\bibliography{ANoteRefs}

\onecolumngrid

\newpage

\begin{appendices}

\section{Implicit Ticking Clock Representation}\label{Sec:Implicit Ticking Clock Representation}

In this \app{} we will formulate a representation of the ticking clock channel $\cM^{t}_{\cl\reT \to \cl \reT}$ which holds if and only if the ticking clock satisfies the axiomatic \cref{def_quantumclock}, up to some stated equivalence. Unlike the  representation of the ticking clock channel from \cref{sec:Explicit Ticking Clock Representation}, this  representation is technical in nature. Its presentation  is followed by some technical implications such as how its delay functions are related to those of~\cite{woods2018quantum}.\Mspace

We start with the following lemma which asserts that the ticking clock from \cref{sec:new quantum clock def} can equivalently be specified in terms of generators acting on the \clock{} space $\mathcal{H}_\cl$.

\begin{restatable}[Implicit ticking clock representation]{lemma}{clockgenerators}\label{lem_clockgenerators}
	The pair $(\rho_{\cl\reT}^0, (\cM^{t}_{\cl \reT \to \cl \reT})_{t\geq 0})$ form a ticking clock (\cref{def_quantumclock}) with a classical tick register (\cref{def:classical register}), up to \clock{}  equivalence (\cref{def:equivalence}), if and only if there exists a Hermitian operator $H$ as well as two finite sequences of operators $(L_j)_{j=1}^\m$ and $(J_j)_{j=1}^\m$ on $\bounded(\mathcal{H}_\cl)$; which are all $k$ and $t$ independent, such that for all $t\geq 0$ and $k=0,1,\ldots,N_T$;
\begin{subequations}
	\begin{align}\label{eq:lemma 1 stement channel at time}
	\begin{split}
	\cM^{t,k}_{\cl \to \cl \reT}(\rho_\cl^0)
	&=
	\! \lim_{\substack{N\to+\infty \\ N\in\nn}} \! \left( \cM^{t/N}_{\cl\reT \to \cl \reT} \right)^{\!\comP(N-1)}\!\! \circ \cM^{ t/N, k}_{\cl \to \cl \reT}(\rho_\cl^0),
	\end{split}
	\end{align}	
	where
	\begin{align}
	\cM^{t/N,\,k}_{\cl \to \cl \reT}(\cdot)
	= \,&(\cdot) \otimes \proj{k}_\reT + \left(\!\frac{t}{N}\!\right)\mathcal{C}_{(1,k)}(\cdot) \otimes\proj{k}_\reT
+ \left(\!\frac{t}{N}\!\right)\mathcal{C}_{(2,k)}(\cdot) \otimes \proj{k\!+\!1}_\reT+F^{t/N,\,k}_{\cl \to \cl \reT}(\cdot),\label{eq:M 3 term}
	\end{align}
	with
	\begin{align}
	\mathcal{C}_{(1,k)}(\cdot)&:=  - \mi [H, \cdot]  - \sum_{j=1}^\m \frac{1}{2} \{L^{\dagger}_j L_j + \theta(k)J^{\dagger}_j J_j, \cdot\} + L_j (\cdot) L_j^{\dagger},\label{eq:C 1 l}\\
	\mathcal{C}_{(2,k)}(\cdot)&:= \theta(k)\sum_{j=1}^\m J_j (\cdot) J_j^{\dagger},\label{eq:C 2 l}
	\end{align}
\end{subequations}
	and $F^{\delta,k}_{\cl \to \cl \reT}(\rho_\cl^0) =  \lo\,(\delta)$ entry-wise.  $\theta(k)=1$ for all $k$ in the periodic register case and $\theta(k)=1-\hat\delta_{k,N_T}$ in the cut-off register case, where $\hat\delta_{\cdot,\cdot}$ is the Kronecker delta.
\end{restatable}
The proof of this lemma, which uses some elements of the proof of Lindblad's representation theorem~\cite{Lindblad, Vittorio}, is provided in \App{} \cref{sec_GeneratorLemmaProof}.\Mspace

As regards to the dynamics on $\mathcal{H}_\cl\otimes\mathcal{H}_\reT$, \cref{eq:lemma 1 stement channel at time} is completely determined by \cref{eq:M 3 term} in terms of the initial state and operators $H$,  $(L_j)_j$, $(J_j)_j$. To see this, 1st note that applying the channel $\cM^{t/N}_{\cl\reT \to \cl \reT}$ to both sides of \cref{eq:M 3 term} one obtains
the composition law
	\begin{align}\label{eq:recursion relation in lemma}
	\cM^{t/N}_{\cl\reT \to \cl \reT} \circ  \cM^{t/N, l}_{\cl \to \cl \reT}\left(\cdot\right)&=
	\cM^{t/N, l}_{\cl \to \cl \reT}\left(\cdot\right) +\left(t/N\right)\,  \cM^{t/N, l}_{\cl \to \cl \reT}\left(\mathcal{C}_{(1,l)}(\cdot)\right)+ \left(t/N\right)\,  \cM^{t/N, l+1}_{\cl \to \cl \reT}\left(\mathcal{C}_{(2,l)}(\cdot)\right)+  F^{\delta,k}_{\cl \to \cl \reT}(\cdot),
	\end{align}
with $l=0,1,\ldots, N_T$. This establishes \cref{eq:lemma 1 stement channel at time} inductively. Every one of the $N$ applications of the channel $\cM^{t/N}_{\cl\reT \to \cl \reT}$ in \cref{eq:lemma 1 stement channel at time}, has a direct physical meaning. Up to order $\lo(t/N)$ contributions, the $i^\text{th}$ application of  $\cM^{t/N}_{\cl\reT \to \cl \reT}$ takes the state from the previous time step \big($i-1$ applications of $\cM^{t/N}_{\cl\reT \to \cl \reT}$\big), and updates the state and probability so that if $l$ ticks had occurred in the previous time step, then either no tick occurs or the ticking clock ticks once, in the $i^\text{th}$ time step. Other processes, such as the clock ``loosing a tick'' or ticking more than once in the time step $t/N$ can only occur with probability $\lo(t/N)$. However, when taking the $N\to+\infty$ limit in \cref{eq:lemma 1 stement channel at time} the order $\lo(t/N)$ terms vanish. As such they are irrelevant and can be set to zero if one wishes. Furthermore, the requirement that $F^{\delta,k}_{\cl \to \cl \reT}(\rho_\cl) =  \lo\,(\delta)$ entry-wise, holds if and only if $\| F^{\delta,k}_{\cl \to \cl \reT}(\rho_\cl)\|_p =  \lo\,(\delta)$ for any $p> 0$ where $\| \cdot\|_p$ is the operator norm induced by the vector $p$-norm. This is shown in \cref{lem: p norm equiv} in the \app.\Mspace

Observe that we have not placed any restrictions on $\m\in\nn$. It turns out that without loss of generality, one can set $\m=d^2-1$, where $d$ is the Hilbert space dimension of the \clock. This is because for any two finite sequences $(L_j)_{j=1}^\m,(J_j)_{j=1}^\m$ giving rise to \cref{eq:lemma 1 stement channel at time}, there exists a new set $(L_j')_{j=1}^{d^2-1},(J_j')_{j=1}^{d^2-1}$ which gives rise to the same dynamics in \cref{eq:lemma 1 stement channel at time} upon their substitution. This follows from simple variants of well known proofs in quantum channel representation theory, as shown in \cref{lem: up bound on m} in the \app. Since the representations of the channel $\cM^{t}_{\cl\reT \to \cl \reT}$ in \cref{lem_clockgenerators,prop:autonmous Lindbald form} are {equivalent}, the choice $N_L=d^2-1$ can also be made in \cref{prop:autonmous Lindbald form} w.l.o.g.\Mspace

If in the definition of a ticking clock (\cref{def_quantumclock}), we remove the condition that the channel from $\bounded(\mathcal{H}_\cl\otimes\mathcal{H}_\reT)$ to $\bounded(\mathcal{H}_\cl\otimes\mathcal{H}_\reT)$ is Markovian [\cref{eq_extension CR to CR}], then \cref{lem_clockgenerators} still holds if one traces out the register in both sides of \cref{eq:lemma 1 stement channel at time} to produce the channel $\cM_{\cl\to\cl}^t$. 
In such cases, it appears that the full channel $\cM_{\cl\reT\to\cl\reT}^t$ producing the dynamics on the register is undetermined other than for an infinitesimal time step. While such channels do not allow one to determine all properties of the clock, it does allow one to determine the probability of the ticking clock ``ticking" during infinitesimal time step $[t,t+dt]$. Denoting this probability $P_\textup{tick}\left(\rho_\cl(t)\right) dt$, one has that $P_\textup{tick}\left(\rho_\cl(t)\right)= \sum_{k=0}^{N_T} p_k(t) P_\textup{tick}^{(k)}\left(\rho_\cl(t)\right)$ where $P_\textup{tick}^{(k)}\left(\rho_\cl(t)\right)$ is the probability density corresponding to ticking during coordinate time interval $[t,t+dt]$, given that the probability of the \clock{} and register being in state $\rho_\cl(t)\otimes\proj{k}_\reT$ at time $t$, was $p_k(t)$. For all probability distributions $\big(p_k(t)\big)_k$ it takes the value     \Mspace
\begin{align}\label{eq:ticking prob}
 P_\textup{tick}\!\left(\rho_\cl(t)\right) &= \sum_{k=0}^{N_T} p_k(t)  \lim_{\delta\rightarrow 0^+} \frac{ \tr \Big[\! \proj{k\!+\!1}_\reT\! \left(\cM_{\cl \to \cl\reT}^{\delta,k} (\rho_\cl(t))-\rho_\cl(t)\otimes\proj{k}_\reT\right) \!\Big]\! }{\delta}\nonumber\\
&= \left(\sum_{k=0}^{N_T} p_k(t)\,\theta(k)\right)\tr\Bigg[ \Bigg(\sum_{j=1}^\m  J_j^\dag J_j \Bigg)\rho_\cl(t)  \Bigg].
\end{align}
where $\rho_\cl(t)=\cM_{\cl\to\cl}^t(\rho_\cl)$ is the \clock{} state at coordinate time $t$. Since $\sum_{k=0}^{N_T} p_k(t)=1$, in the case of the periodic register, the factor $\sum_{k=0}^{N_T} p_k(t)\,\theta(k)$ vanishes from \cref{eq:ticking prob} and $P_\textup{tick}\!\left(\rho_\cl(t)\right)$ only depends on the local \clock{} dynamics. 
However, such probabilities are not so useful for determining measures of ticking clock accuracy, since they do not provide information about individual ticks.\Mspace

For example, a more useful quantity is the probability of producing the $k^\textup{th}$ tick during coordinate time interval $[t,t+dt]$. Or in other words, the probability that during time interval $[0,t]$, the ticking clock ticked $k-1$ times and then produced a tick during time interval $[t,t+dt]$. We denote this probability measure $P_\textup{ticks}^{(k)}(t) \,dt$, where $P_\textup{ticks}^{(k)}(t)$ is called the \emph{$k^\text{th}$ tick delay function}. One finds
\begin{align}\label{eq:ticking from state k prob}
P_{\textup{ticks}}^{(k)}(t)
=\! \lim_{\delta\rightarrow 0^+}
&\frac{ \tr \Big[\! \proj{k}_\reT\! \left(\cM_{\cl\reT \to \cl\reT}^{\delta} (\rho^{(k-1)}_{\cl\reT}(t))-\rho_{\cl\reT}^{(k-1)}(t)\right)\! \Big]\! }{\delta},
\end{align}
where $\rho^{(k-1)}_{\cl\reT}(t)$ is the un-normalised outcome of a measurement when the register is found to be in the $\proj{k\!-\!1}_\reT$ state,
\begin{align}
\begin{split}
\rho^{(k-1)}_{\cl\reT}(t)=
&  \proj{k\!-\!1}_\reT\left(\cM_{\cl \to \cl\reT}^{t,0} (\rho_{\cl}^0)\right)\proj{k\!-\!1}_\reT,
\end{split}
\end{align} 
and we have assumed that the \clock{} and register are initialised to $\rho_\cl^0\otimes \proj{0}_\reT$ at coordinate time $t=0$. 
In the case of a cut-off register, from \cref{eq:M 3 term}, we observe that in every infinitesimal time step, the dynamics incurred on the \clock{} for the first $N_T-1$ ticks is identical to that derived in \cite{woods2018quantum} for the model \cite{RaLiRe15} in the continuous time limit. 
Therefore, the $k^\text{th}$ tick delay function, $P_{\textup{ticks}}^{(k)}(t)$, is the same for the ticking clocks in \cref{prop:autonmous Lindbald form} and those in \cite{woods2018quantum} for $k=1,2,3, \ldots, N_T$ in the cut-off register case. This has important consequences as discussed in \cref{sec:accuracy}.\Mspace

For example, consider the case of the 1st tick in the case of the cut-off register model. A simple calculation finds that $\rho^{(0)}_{\cl}(t)$ is generated via the \clock's dynamics with the tick generating channel removed:
\begin{align}
\rho^{(0)}_{\cl}(t) &=\tr_\reT\left[ \proj{0}_\reT\left(\cM_{\cl \to \cl\reT}^{t,0} (\rho_{\cl})\right) \right]\,\\
&= \!\lim_{N\to +\infty}\!\tr_\reT\!\left[ \left(\proj{0}_\reT\cM_{\cl \to \cl\reT}^{t/N,0} (\rho_{\cl})\proj{0}_\reT\right)^{\!\comP N} \right]\label{eq:eqfor COM}\\
&= \me^{t \mathcal{L}_{\cl}^{(0)}} (\rho_\cl),\\
\mathcal{L}_{\cl}^{(0)} \left( \cdot \right) &= \mathcal{C}_{(1,0)}(\cdot).
\end{align}
In the case of the periodic register model, the above expression for $\rho^{(0)}_{\cl}(t)$ does not hold, since the equality in line \eqref{eq:eqfor COM} is false. Physically speaking, this is because in the periodic register case, when the register runs out of memory, a tick is produced in the initial memory state $\proj{0}_\reT$ and thus the dynamics of the clock at times after the 1st tick has occurred are still relevant for the 1st tick's statistics. This is not the case in the infinite register limit $N_T\to+\infty$ nor for the cut-off register case\fnr{ or when the register never runs out of memory (i.e.\ $N_T=\infty$)}.\Mspace	

\section{Proofs}

\subsection{Proof of \cref{prop:continous measurements}}\label{sec:measurement equiavalence}

\continuousMeasurements*

\begin{proof}
	To start with, observe that for a ticking clock $(\rho_\cl^0, \cM^{t}_{\cl\reT \to \cl \reT})$ with a classical register, one has for all $t\geq 0$; $k=0,1,\ldots, N_T$; $\rho_\cl^0\in\mathcal{S}(\mathcal{H}_\cl)$, 
	\begin{align}
	\cM^{t,k}_{\cl \to \cl \reT}(\rho_\cl^0)=& \sum_{l=0}^{N_T} \tr_\reT\left[ \proj{l}_\reT \cM^{t,k}_{\cl \to \cl \reT}(\rho_\cl^0) \right] \otimes \proj{l}_\reT\\
	{=}& \sum_{l=0}^{N_T} \mathtt{Prob}[l,t]\, \mathcal{P}_l\circ \cM^{t,k}_{\cl\to \cl \reT}(\rho_\cl^0),
	\end{align}
	where $\mathtt{Prob}[l,t]:= \tr\left[ \proj{l}_\reT \cM^{t,k}_{\cl\reT \to \cl \reT}(\rho_\cl^0) \right]$. Thus, using the Markovian property of channel $\cM^{t}_{\cl\reT \to \cl \reT}$ [condition 2)], followed by iteratively substituting the above equation,
	\begin{align}
	&\cM^{t,k}_{\cl \to \cl \reT}(\rho_\cl^0)\\
	&= \circsum{1}{N} \cM^{t_n}_{\cl\reT \to \cl \reT}\left(\rho_\cl^0 \otimes \proj{k}_\reT\right)\\
	&= \circsum{2}{N} \cM^{t_n}_{\cl\reT \to \cl \reT}\circ \sum_{l_1=0}^{N_T} \mathtt{Prob}[l_1,t_1]\, \mathcal{P}_{l_1}\circ \cM^{t_1}_{\cl\reT \to \cl \reT} \left(\rho_\cl^0 \otimes \proj{k}_\reT\right)\\
	&= \circsum{3}{N} \cM^{t_n}_{\cl\reT \to \cl \reT}\circ \sum_{l_2=0}^{N_T} \mathtt{Prob}[l_2,t_2]\, \mathcal{P}_{l_2}\circ \cM^{t_2}_{\cl\reT \to \cl \reT} \circ\sum_{l_1=0}^{N_T} \mathtt{Prob}[l_1,t_1]\, \mathcal{P}_{l_1}\circ \cM^{t_1}_{\cl\reT \to \cl \reT} \left(\rho_\cl^0 \otimes \proj{k}_\reT\right)\\
	&= \sum_{l_1=0}^{N_T}\ldots \sum_{l_N=0}^{N_T} \mathtt{Prob}[l_1,t_1]\cdots \mathtt{Prob}[l_N,t_N] \circsum{1}{N} \mathcal{P}_{l_n}\circ \cM^{t_n}_{\cl\reT \to \cl \reT} \left(\rho_\cl^0 \otimes \proj{k}_\reT\right).
	\end{align}
	Observe that $\mathtt{Prob}[l_1,t_1]\cdots \mathtt{Prob}[l_N,t_N]= \mathrm{Prob}[{\vec s}_N]$. Taking into account \cref{def:continous measured channel} we complete the proof.
\end{proof}

\subsection{Proof of \cref{lem_clockgenerators}} \label{sec_GeneratorLemmaProof}
\clockgenerators*
\begin{proof}
First we will prove the proposition for the case of a ticking clock with a periodic register. The case of a cut-off register will then be straightforward. To start with, we consider the most general representation of a channel in the Kraus form, namely
\begin{align}
\cM_{\cl  \to \cl  \reT}^{t,k} (\rho_{\cl }) = \sum_{j=0}^{\NQ } Q_j^{(k)}(t) \rho_\cl\, {Q_j^{(k)}}^\dag(t),  
\end{align}	
where 
\begin{align}\label{eq:normalisation Krouas ops}
Q_j^{(k)}(t): \bounded(\mathcal{H}_\cl)\to\bounded(\mathcal{H}_\cl\otimes\mathcal{H}_\reT)
\end{align}
and $\sum_j {Q_j^{(k)}(t)}^\dag Q_j^{(k)}(t)=\id_\cl$. We now expand the $Q_j^{(k)}(t)$ operators in the register basis. By making the identification $N_j^{(k)}(l,t):= \bra{l}_\reT Q_j^{(k)}(t): \bounded(\mathcal{H}_\cl)\to\bounded(\mathcal{H}_\cl)$ this yields
\begin{align}
\cM_{\cl  \to \cl  \reT}^{t,k} (\rho_{\cl }) = \sum_{l,l'=0}^{N_T}\sum_{j=0}^{\NQ } N_j^{(k)}(l,t) \rho_\cl\, {N_j^{(k)}}^\dag(l',t)\otimes\ketbra{l}{l'}_\reT,  
\end{align}	
since the register is classical, the off diagonal terms must vanish due to compatibility with \cref{eq:classical reg def}. We therefore have
\begin{align}\label{eq:classical reg def 2}
\cM_{\cl  \to \cl  \reT}^{t,k} (\rho_{\cl }) = \sum_{l=0}^{N_T}\sum_{j=0}^{\NQ } N_j^{(k)}(l,t) \rho_\cl\, {N_j^{(k)}}^\dag(l,t)\otimes\ketbra{l}{l}_\reT,  
\end{align}	
for all $ k=0,1,\ldots,N_T;\, t\geq 0$. Observe that
\begin{align}\label{}
\tr_\reT\left[\cM^{t,k}_{\cl \to \cl \reT}(\rho_\cl)\proj{l}_\reT\right] =\sum_{j=0}^{\NQ }N_j^{(k)}(l,t) \rho_\cl\, {N_j^{(k)}}^\dag(l,t).
\end{align}
Moreover, in the periodic register case, $\ket{l}_\reT=\ket{l \textup{ mod. } N_T+1 }_\reT$. For convenience, we therefore extend the definition of the operators $N_j^{(k)}(l,t)$ in the periodic case from $l=0,1,\ldots N_T$ to $l\in\zz$ by defining 
\begin{align}\label{eq:extension periodoc case}
N_j^{(k)}(l,t)= N_j^{(k)}(l \textup{ mod. } N_T+1,t)
\end{align}
 Consider an expansion of the form $N_j^{(k)}(l,t)=\sum_{n,m}a_{n,m}^{(j,k)}(l,t) \ketbra{n}{m}_\cl$ and states $\rho_\cl(p)$ which are pure and diagonal in this basis, $\rho_\cl(p)=\proj{p}_\cl$. Therefore
\begin{align}\label{eq:C to CT gen in proof 3}
  \frac{\tr_\reT\!\left[\cM^{t,k}_{\cl \to \cl \reT}(\rho_\cl(p))\proj{l}_\reT\right]}{\tilde p^{(k)}_{l}(t)} = \sum_{j=0}^{\NQ }\frac{ \sum_{\substack{n,m \\ n\neq m}}  a_{n,p}^{(j,k)}(l,t) {a_{m,p}^{(j,k)}}^*(l,t) \ketbra{n}{m}}{\tilde p^{(k)}_{l}(t)}+ \frac{\sum_{n}  a_{n,p}^{(j,k)}(l,t) {a_{n,p}^{(j,k)}}^*(l,t) \ketbra{n}{n}}{\tilde p^{(k)}_{l}(t)}.
  \end{align}
Now noting the definition of $\tilde p^{(k)}_l(t)$ (\cref{eq:tilde p k l delta}) and taking the trace on both sides, we find
  \begin{align}
  1=  \frac{\sum_{n,j}  |a_{n,p}^{(j,k)}(l,t)|^2 }{\tilde p^{(k)}_{l}(t)}.
  \end{align}
 Therefore,
  \begin{align}\label{eq:coefs constrinat}
  \left|a_{n,p}^{(j,k)}(l,t)\right|\leq \sqrt{\tilde p^{(k)}_{l}(t)}\quad  \forall\, j=0,1,\ldots,\NQ ; \, n, p =0,1,\ldots,d-1; l,k=0,1,\ldots,N_T;\,  t\geq 0,
  \end{align}
 where $d\in\nnp$ is the dimension of the Hilbert space of the \clock.
We will now use inequality \cref{eq:coefs constrinat} together with condition 4) [\cref{eq:leaidng order condition}] to show an important limit. To start with, denote the entries of a matrix $M$ by $[M]_{[ab]}$ and observe that \cref{eq:coefs constrinat} implies
\begin{align}\label{eq:little o scaling lidblad}
\begin{split}
\lim_{t\to 0^+} \left|\frac{\sum_{\underset{l\notin  \{k,  f(k)\}}{l=0\qquad}}^{N_T}\sum_{j=0}^{\NQ }\left[N_j^{(k)}(l,t) \rho_\cl {N_j^{(k)}}^{\dagger}(l,t)\right]_{[ab]}}{\tilde p_{f(k)}^{(k)}(t)}\right|&= \lim_{t\to 0^+} \left|\frac{\sum_{\underset{l\notin  \{k,  f(k)\}}{l=0\qquad}}^{N_T}\sum_{j=0}^{\NQ }\sum_{m,n}a^{(j,k)}_{a,m}(l,t)\left[\rho_\cl\right]_{[m,n]} {a^{(j,k)}_{b,n}}^*(l,t)}{\tilde p_{f(k)}^{(k)}(t)}\right|\\
& \leq (\NQ +1)\left(\sum_{m,n}\left|\left[\rho_\cl\right]_{[m,n]}\right|\right) \lim_{t\to 0^+} \frac{\sum_{\underset{l\notin  \{k,  f(k)\}}{l=0\qquad}}^{N_T} \tilde p_{l}^{(k)}(t)}{\tilde p_{f(k)}^{(k)}(t)}=0,
\end{split}
\end{align}
for all $k=0,1,\ldots,N_T$; $a,b=0,1,\ldots,d-1$;  $\rho_\cl\in\mathcal{S}\left(\mathcal{H}_\cl\right)$. In the last line in \cref{eq:little o scaling lidblad} we have used \cref{eq:leaidng order condition}. Therefore, 
\begin{align}
\sum_{\underset{l\notin  \{k,  f(k)\}}{l=0}}^{N_T}\sum_{j=0}^{\NQ }N_j^{(k)}(l,t) \rho_\cl {N_j^{(k)}}^{\dagger}(l,t)= \lo\left(\tilde p_{f(k)}^{(k)}(t)\right)
\end{align}
entry-wise in the $t\to0^+$ limit for all $k=0,1,\ldots,N_T$. 
Since every term $N_j^{(k)}(l,t) \rho_\cl {N_j^{(k)}}^{\dagger}(l,t)$ in the above summation is positive semi-definite, we thus have
\begin{align}
\cM_{\cl  \to \cl  \reT}^{\delta t,k} (\rho_{\cl }) =& \sum_{l\in \{k,f(k)\}} \sum_{j=0}^{\NQ }N_j^{(k)}(l,\delta t) \rho_\cl\, {N_j^{(k)}}^\dag(l,\delta t)\otimes\ketbra{l}{l}_\reT + \lo\left(\tilde p_{f(k)}^{(k)}(\delta t)\right),\\
=& \sum_{l\in \{0,f(k)-k\}} \sum_{j=0}^{\NQ }N_j^{(k)}(l+k,\delta t) \rho_\cl\, {N_j^{(k)}}^\dag(l+k,\delta t)\otimes\ketbra{l+k}{l+k}_\reT + \lo\left(\tilde p_{f(k)}^{(k)}(\delta t)\right),\label{eq: cl to cl reT}
\end{align}
for all $k=0,1,\ldots,N_T$; $t\geq 0$; $\rho_\cl\in\mathcal{S}\left(\mathcal{H}_\cl\right)$. 
On the other hand, from \cref{eq:classical reg def 2} it follows that the state of the \clock, given the register is measured to be in the state $\proj{k+l}_\reT$ for $l\in\zz$, is
\begin{align}\label{eq:C to CT gen in proof 5}
\tr_\reT\left[\cM^{t,k}_{\cl \to \cl \reT}(\rho_\cl)\proj{k+l}_\reT\right] =\sum_{j=0}^{\NQ }N_j^{(k)}(l+k,t) \rho_\cl\, {N_j^{(k)}}^\dag(l+k,t).
\end{align}
By virtue of condition 1) (\cref{eq:clock invariance}), we have that every matrix component of the r.h.s. of \cref{eq:C to CT gen in proof 5} is $k$ independent for all $t\geq 0$, $\rho_\cl \in\mathcal{S}(\mathcal{H_\cl })$ and $l\in\zz$; $k=0,1,\ldots,N_T$ s.t. $k+l=0,1,\ldots, N_T$ in the periodic register case. Therefore, in particular, we have 
\begin{align}\label{eq:set k to zero}
\sum_{j=0}^{\NQ } N_j^{(k)}(l+k,t) \rho_\cl\, {N_j^{(k)}}^\dag(l+k,t)= \sum_{j=0}^{\NQ } N_j^{(0)}(l,t) \rho_\cl\, {N_j^{(0)}}^\dag(l,t)
\end{align} for all $t\geq 0$, $\rho_\cl \in\mathcal{S}(\mathcal{H_\cl })$ and $l\in\zz$; $k=0,1,\ldots,N_T$ s.t. $k+l=0,1,\ldots, N_T$ in the periodic case. Plugging \cref{eq:set k to zero} into \cref{eq: cl to cl reT}, yields
\begin{align}\label{eq: cl to cl reT 2}
\cM_{\cl  \to \cl  \reT}^{\delta t,k} (\rho_{\cl }) = \sum_{l\in \{0,f(k)-k\}} \sum_{j=0}^{\NQ }N_j^{(0)}(l,\delta t) \rho_\cl\, {N_j^{(0)}}^\dag(l,\delta t)\otimes\ketbra{l+k}{l+k}_\reT + \lo\left(\tilde p_{f(k)}^{(k)}(\delta t)\right),
\end{align}
for all $k=0,1,\ldots,N_T-1$; $\rho_\cl\in\mathcal{S}\left(\mathcal{H}_\cl\right)$ in the periodic register case. For the case $k=N_T$ in the above equation, recall that due to \cref{eq:extension periodoc case} we have that 
\begin{align}
\sum_{j=0}^{\NQ } N_j^{(0)}(-N_T,t) \rho_\cl\, {N_j^{(0)}}^\dag(-N_T,t)= \sum_{j=0}^{\NQ } N_j^{(0)}(1,t) \rho_\cl\, {N_j^{(0)}}^\dag(1,t).\label{eq:set k to zero 2}
\end{align}

For the periodic register case, taking into account \cref{eq:set k to zero,eq:set k to zero 2} we have
\begin{align}\label{eq:p order invar periodic reg}
\tilde p_{f(k)}^{(k)}(t)= \tr\left[\sum_{j=0}^{\NQ } N_j^{(k)}(k+1,t)\rho_\cl {N_j^{(k)}}^\dag(k+1,t)\right]= \tr\left[\sum_{j=0}^{\NQ } N_j^{(0)}(1,t)\rho_\cl {N_j^{(0)}}^\dag(1,t)\right]= \tilde p_{1}^{(0)}(t),
\end{align}
for all $k=0,1,\ldots, N_T$; $t\geq 0$; $\rho_\cl\in\mathcal{S}\left(\mathcal{H}_\cl\right)$. Therefore, for all  $k=0,1,\ldots, N_T$; $t\geq 0$; $\rho_\cl\in\mathcal{S}\left(\mathcal{H}_\cl\right)$;  \cref{eq: cl to cl reT 2} reduces to
\begin{align}\label{eq: cl to cl reT 13}
\cM_{\cl  \to \cl  \reT}^{\delta t,k} (\rho_{\cl }) = \sum_{l\in \{0,1\}} \sum_{j=0}^{\NQ }N_j^{(0)}(l,\delta t) \rho_\cl\, {N_j^{(0)}}^\dag(l,\delta t)\otimes\ketbra{l+k}{l+k}_\reT + \lo\left(\tilde p_{1}^{(0)}(\delta t)\right).
\end{align}

It follows from \cref{eq:id channel clockwork a,eq:id channel clockwork b,eq_extension CR to CR} of condition 3), that the \clock{} channel admits a power-law expansion in $t$ (see uniformly continuous semigroup in~\cite{Pazy1983}). Specifically, there exits an operator $A_{\cl\reT}$ on $\bounded(\mathcal{H}_\cl\otimes\mathcal{H}_\reT)$ such that
\begin{align}
\cM_{\cl\reT  \to \cl\reT }^{t}=\me^{t A_{\cl\reT}}:=\sum_{n=0}^\infty t^n\frac{ A_{\cl\reT}^{\comP n}}{n!}.
\end{align}
Therefore, w.l.o.g. we can use the following ansatz: 

Let $n_1\leq \NQ $ of the elements of the sequence $\big(\,N_j^{(0)}(0,\delta t)\,\big)_j$ be of linear order in $\delta t$ while the others are of order  $\sqrt{\delta t}$. Specifically, let
\begin{align}\label{eq:N j 0 1st}
\Big (\, N_j^{(0)}(0,\delta t)=I_j+(-\mi H_j+K_j)\delta t\, \,\Big)_{j=0}^{n_1},\quad \Big (\, N_j^{(0)}(0,\delta t)=L_j \sqrt{\delta t}\, \,\Big)_{j=n_1+1}^{\NQ },
\end{align} 
where $H_j$, $K_j$ are Hermitian and the operators $L_j$, $H_j$, $K_j$, $I_j$, are all $t$ independent. Similarly, we can employ the same form of the expansion for the sequence $\big(\,N_j^{(0)}(1,\delta t)\,\big)_j$ associated with the register state $\ketbra{k+1}{k+1}_\reT:$
\begin{align}\label{eq:N j 1 1st}
\Big (\, N_j^{(0)}(1,\delta t)=I_j'+(-\mi \bar H_j+\bar K_j)\delta t\, \,\Big)_{j=0}^{n_1},\quad \Big(\, N_j^{(0)}(1,\delta t)=J_j \sqrt{\delta t}\, \,\Big)_{j=n_1+1}^{\NQ },
\end{align} 
where $\bar H_j$, $\bar K_j$ are Hermitian and the operators $J_j$, $\bar H_j$, $\bar K_j$, $I_j'$, are all $t$ independent operators.

We first fix the zeroth order terms by noting that $\cM_{\cl  \to \cl}^{0}:=\tr_\reT\left[\cM_{\cl  \to \cl\reT}^{0,k}\right]$ has to be the identity channel due to condition 3) [\cref{eq:id channel clockwork a}] and \cref{eq: cl to cl reT 13}. It hence follows:
\begin{align}\label{eq:M c to c id}
\cM_{\cl  \to \cl}^{0} (\rho_{\cl})=\sum_{j=0}^{n_1} I_j\rho_\cl I_j^\dag+ I_j'\rho_\cl{I_j'}^\dag=\rho_\cl
\end{align}
for all $\rho_\cl\in\mathcal{S}\left(\mathcal{H}_\cl\right)$. Two sets of Kraus operators $(\tilde K_l)_{l=0}^{n_1}$, $(\tilde K_l')_{l=0}^{n_1}$ give rise to the same quantum channel (i.e.\ $\sum_l \tilde K_l'(\cdot)\tilde {K_l'}^\dag=\sum_l \tilde K_l(\cdot)\tilde {K_l}^\dag$\,) iff there exists an $n_1$ by $n_1$ unitary $V$ with entries $V_{[lm]}\in\cc$ such that $\tilde K'_l=\sum_m V_{[lm]} \tilde K_m$ for all $l=0,1,\ldots,n_1$. See e.g.~\cite{Nilson,kraus} for a proof. Note that this even covers the case in which one or both of the channels have less than $n_1+1$ Kraus operator elements, since we can always choose to define additional Kraus operators which are equal to zero. Therefore, since $(I_j=c_j \id, I_j'=c_j'\id)_{j=0}^{n_1}$ with $\sum_{j=0}^{n_1} |c_j|^2+|c_j'|^2=1$ are solutions to \cref{eq:M c to c id}, this unitary equivalence theorem implies that it is the only family of solutions. On the other hand, one also finds 
\begin{align}
\lim_{t\to 0^+} \tilde p_0^{(0)}(t)&= \sum_{j=0}^{n_1}\tr\bigg[I_j\rho_\cl I_j^\dag \bigg]=1,\label{eq:tilde p k zeroth}\\
\lim_{t\to 0^+} \tilde p_{1}^{(0)}(t)&= \sum_{j=0}^{n_1}\tr\bigg[I_j'\rho_\cl {I_j'}^\dag \bigg]=0,\label{eq:tilde p k+1 zeroth}
\end{align}
for all $\rho_\cl\in\mathcal{S}\left(\mathcal{H}_\cl\right)$. The last equality in \cref{eq:tilde p k zeroth} follows from invoking condition 3) [\cref{eq:id channel clockwork b}] and the definition of $\tilde p_0^{(0)}$, while the last equality in \cref{eq:tilde p k+1 zeroth} follows from conservation of probability. Since $I_j'\rho_\cl {I_j'}^\dag$ is positive semi-definite, it follows that $I_j'=0$ for all $j=0,1,\ldots,n_1$ and thus we find $I_j=c_j\id$ with $\sum_{j=0}^{n_1} |c_j|^2=1$  for all $j=0,1,\ldots,n_1$. It thus follows from plugging in the above ansatz for $N_j^{(0)}(0,t)$ and $N_j^{(0)}(1,t)$ into the definition of $\tilde p_{1}^{(0)}$ that 
\begin{align}\label{eq:prob of order delta t}
\lo\big(\tilde p_{1}^{(0)}(\delta t)\big)&=\lo(\delta t).
\end{align}
Expanding \cref{eq: cl to cl reT 13} we thus find
\begin{align}\label{eq:lidbald 3rd}
\begin{split}
\cM_{\cl  \to \cl}^{\delta t,k} (\rho_{\cl}) =&\, \rho_{\cl}\otimes \proj{k}_\reT+  \bigg[\rho_\cl, \sum_{j=0}^{n_1}\left(c_j^\textup{R} H_j+c_j^\textup{I} K_j\right) \bigg]\otimes \proj{k}_\reT \delta t+ \bigg\{\rho_\cl, \sum_{j=0}^{n_1}\left(c_j^\textup{I} H_j+c_j^\textup{R} K_j\right) \bigg\}\otimes \proj{k}_\reT\delta t\\
& + \sum_{j=n_1+1}^{\NQ }\left(L_j\rho_{\cl}L_j^\dag\otimes \proj{k}_\reT+J_j\rho_{\cl}J_j^\dag\otimes \proj{k+1}_\reT\right)\delta t + \lo(\delta t),
\end{split}
\end{align}
where $c_j^\textup{R}$, $c_j^\textup{I}$ are the real and imaginary parts of $c_j$ respectively. Now observe that by defining Kraus operators $H:= \sum_{j=0}^{n_1}\left(c_j^\textup{R} H_j+c_j^\textup{I} K_j\right)$, $K:=  \sum_{j=0}^{n_1}\left(c_j^\textup{I} H_j+c_j^\textup{R} K_j\right)$ one can exchange \cref{eq:N j 0 1st,eq:N j 1 1st} with
\begin{equation}\label{eq:N j 0 2nd}
\begin{aligned}
&N_0^{(0)}(0,\delta t) =\id+(-\mi H+K)\delta t,\quad &\Big (\, N_j^{(0)}(0,\delta t)=L_j \sqrt{\delta t}\, \,\Big)_{j=1}^{\NQ },&\\
&N_0^{(0)}(1,\delta t)=0,\quad &\Big (\, N_j^{(0)}(1,\delta t)=J_j \sqrt{\delta t}\, \,\Big)_{j=1}^{\NQ },&
\end{aligned}
\end{equation}
and obtain the same solution as \cref{eq:lidbald 3rd} to order $\lo(\delta t)$ up to a relabelling of the summation indices obtaining
\begin{align}\label{eq:MC to CT gens}
\begin{split}
\cM^{\delta t,k}_{\cl \to \cl \reT}(\rho_\cl)
= &\,\rho_\cl \otimes \proj{k}_\reT - \delta t \Biggl(  \mi [H, \rho_\cl]  -\{K, \rho_\cl\} - \sum_{j=1}^\m L_j \rho_\cl L_j^{\dagger}\Biggr) \otimes \proj{k}_\reT \\ &+ \,\delta t \sum_{j=1}^\m J_j \rho_\cl J_j^{\dagger} \otimes \proj{k+1}_\reT + \lo\left(\delta t\right)
\end{split}
\end{align}
for all $ k=0,1,\ldots,N_T;\, t\geq 0$; $\rho_\cl\in\mathcal{S}\left(\mathcal{H}_\cl\right)$. Furthermore, taking into account the normalisation of the Kraus operators [\cref{eq:normalisation Krouas ops}], from \cref{eq:N j 0 2nd} we obtain a solution for $K$, namely
\begin{align}
K= -\frac{1}{2} \sum_{j=1}^\m \left(L^{\dagger}_j L_j +  J^{\dagger}_j J_j\right). 
\label{eq: K value}
\end{align}
In the case of a periodic register, \cref{eq:M 3 term} in the lemma statement follows by pugging in \cref{eq: K value} into \cref{eq:MC to CT gens}. \Cref{eq:recursion relation in lemma} in the lemma statement follows by recalling \cref{eq:k channel notation def} and using \cref{eq:MC to CT gens}. By applying the divisibility of the channel [condition 2), \cref{eq_extension CR to CR}] recursively $N\in\nnp$ times we find 
\begin{align}
\left( \cM^{ t/N}_{\cl\reT \to \cl \reT} \right)^{\comP N} = \cM^{ t}_{\cl\reT \to \cl \reT},
\end{align}
for all $t\geq 0$. \Cref{eq:lemma 1 stement channel at time} then follows by recalling the notation \cref{eq:k channel notation def} and taking the limit $N\to+\infty $. This concludes the ``only if'' part of the lemma for a periodic register.\Mspace

Finally, to verify the converse part of the lemma in the case of a periodic register, one simply has to check that for all Hermitian operators $H$ and families of operators $(L_j)_{j=1}^m$, $(J_j)_{j=1}^m$ acting on $\bounded(\mathcal{H}_\cl)$, the equations in the lemma statement satisfy the conditions 1) to 4) in \cref{sec:new quantum clock def}. We do this in the following. For conciseness, we will refer to the sequence of such operators on $\bounded(\mathcal{H}_\cl)$ by $\mathcal{D}= \big(  H, (L_j)_j, (J_j)_j\big)$.

To verify that 1) [\cref{eq:clock invariance}] holds for all $\mathcal{D}$, first note that 
$\tr_\reT\!\left[ \left(\cM^{\tau}_{\cl\reT \to \cl \reT}\right)^{\comP l}\! \circ \cM^{ \tau, k}_{\cl \to \cl \reT}(\rho_\cl)\ketbra{k+m}{k+m}\right]$ is $k$ independent (up to an order $\lo(\tau)$ term) for all $l\in\nnp$, $m\in\zz$, $k=0,1,\ldots, N_T$ s.t. $k+m=0,1,\ldots, N_T$, for all $\tau\geq 0$; $\rho_\cl\in\mathcal{S}\left(\mathcal{H}_\cl\right)$; and $\mathcal{D}$. Hence condition 1) [\cref{eq:clock invariance}] follows by choosing $\tau=t/N$, $l=N-1$ and taking the $N\to\infty$ limit such that the $\lo\left(\tau\right)$ terms vanish.

To verify that \cref{eq_extension CR to CR} in condition 2) holds for all $\mathcal{D}$, one needs to verify that
\begin{align}
\lim_{N\to\infty} \left(\cM^{(t_1+t_2)/N}_{\cl\reT\to\cl\reT}\right)^{\comP N} \left(\rho_\cl\otimes\ketbra{k}{k}\right)
\end{align} and 
\begin{align}
\lim_{N_1\to\infty} \lim_{N_2\to\infty} \left(\cM^{t_1/N_1}_{\cl\reT\to\cl\reT}\right)^{\comP N_1} \circ \left(\cM^{t_2/N_2}_{\cl\reT\to\cl\reT}\right)^{\comP N_2} \left(\rho_\cl\otimes\ketbra{k}{k}\right)\label{eq:lm N in rppof double}
\end{align}
are equal for all $t_1, t_2\geq 0$; $k=0,1,\ldots, N_T$; $\rho_\cl\in\mathcal{S}\left(\mathcal{H}_\cl\right)$; and $\mathcal{D}$.
 To do so, we first note by explicit calculation using \cref{eq:M 3 term} that $ \cM^{t_1/N}_{\cl\reT\to\cl\reT} \circ \cM^{t_2/N}_{\cl\reT\to\cl\reT} \left(\rho_\cl\otimes\ketbra{k}{k}\right)= \cM^{(t_1+t_2)/N}_{\cl\reT\to\cl\reT} \left(\rho_\cl\otimes\ketbra{k}{k}\right)+\lo\left(1/N\right)$. Therefore
\begin{align}
\quad&\lim_{N\to\infty} \left(\cM^{(t_1+t_2)/N}_{\cl\reT\to\cl\reT}\right)^{\comP N} \left(\rho_\cl\otimes\ketbra{k}{k}\right)\\
& = \lim_{N\to\infty} \left(\cM^{t_1/N}_{\cl\reT\to\cl\reT} \circ \cM^{t_2/N}_{\cl\reT\to\cl\reT} +\lo(1/N)\right)^{\comP N} \left(\rho_\cl\otimes\ketbra{k}{k}\right)\\
& = \lim_{N\to\infty} \left(\cM^{t_1/N}_{\cl\reT\to\cl\reT} \circ \cM^{t_2/N}_{\cl\reT\to\cl\reT}\right)^{\comP N}\left(\rho_\cl\otimes\ketbra{k}{k}\right)\\ &\quad\, + N\lo(t/N)+(N-1)\lo(t/N)^2+(N-2)\lo(t/N)^3+\ldots +(N-(N-1))\lo(t/N)^N\\
& = \lim_{N\to\infty} \left(\cM^{t_1/N}_{\cl\reT\to\cl\reT} \circ \cM^{t_2/N}_{\cl\reT\to\cl\reT}\left(\rho_\cl\otimes\ketbra{k}{k}\right)\right)^{\comP N}\\ & \quad\,+ N\lo(t/N)+(N-1)^2\lo(t/N)^2\\
&= \lim_{N\to\infty} \left(\cM^{t_1/N}_{\cl\reT\to\cl\reT}\right)^{\comP N} \circ \left(\cM^{t_2/N}_{\cl\reT\to\cl\reT}\right)^{\comP N} \left(\rho_\cl\otimes\ketbra{k}{k}\right)
\end{align}
which is equal to \cref{eq:lm N in rppof double} due to continuity. 
The confirmation that \cref{eq:id channel clockwork a,eq:id channel clockwork b} in condition 3) hold for all $\mathcal{D}$, follows straightforwardly from \cref{eq:M 3 term}:
\begin{align}
&\lim_{t\to 0^+} \lim_{N\to\infty}   \left(\cM^{t/N}_{\cl\reT\to\cl\reT}
\right)^{\comP N}\left(\rho_\cl\otimes\ketbra{k}{k}\right) \\
&= \lim_{t\to 0^+} \cM^{t}_{\cl\reT\to\cl\reT}\left(\rho_\cl\otimes\ketbra{k}{k}\right) = \rho_\cl\otimes\ketbra{k}{k}, 
\end{align} 
for all $\rho_\cl\in\mathcal{S}\left(\mathcal{H}_\cl\right)$, $\mathcal{D}$  and where in the second equality we used the Markovianity of the channel (which we have just proven) and the penultimate line uses \cref{eq:M 3 term}.

Finally, the verification that condition 4) holds for all $\mathcal{D}$ is straightforward. Plugging in \cref{eq:M 3 term} into definition \cref{eq:tilde p k l delta} and proceeding similarly to as in the above equation, one finds
\begin{align}
\lim_{t\to 0^+} \frac{\sum_{\underset{l\notin  \{k,  f(k)\}}{l=0\qquad}}^{N_T} \tilde p_{l}^{(k)}(t)}{\tilde p_{f(k)}^{(k)}(t)}= \lim_{t\to 0^+} \frac{\lo(t)}{c \,t }=0,
\end{align}
for all $k=0,1,\ldots,N_T$; $\rho_\cl\in\mathcal{S}\left(\mathcal{H}_\cl\right)$; and $\mathcal{D}$,  where $c> 0$ is a constant. This concludes the proof of the converse part of the lemma for the case of a periodic register. We now proceed to the case of a cut-off register.

We have just proven that $\cM_{\cl \reT \to \cl \reT}$ is the channel of a ticking clock with a classical register of the periodic type, iff it has the form stated in the lemma. Therefore a ticking clock with a classical register of the cut-off type can only satisfy the 1st part of condition 5) (\cref{eq:condition 5 a}) iff the ticking clock $\tilde \cM_{\cl \reT \to \cl \reT}$ in \cref{eq:condition 5 a} is of the form of that in the lemma statement for $k=0,1,\ldots, N_T-1$. Since \cref{eq:M 3 term} is the same for both cut-off and periodic register types, for $k=0,1,\ldots, N_T-1$, this holds true. 
Furthermore, by direct calculation of \cref{eq:M 3 term} in the case of a cut-off register and $k=N_T$, we see that it satisfies \cref{eq:condition 5 b} for $k=N_T$. While \cref{eq:M 3 term} in the case of a cut-off register and $k=N_T$ is clearly not necessary for it to satisfy \cref{eq:condition 5 b} for $k=N_T$, it is necessary to satisfy \cref{eq:condition 5 b} for $k=N_T$ up to \clock{} equivalence (\cref{def:equivalence}). This can be verified by noting that \cref{eq:condition 5 b} for $k=N_T$ implies that the state of the register and clockwork must be a product state up to order $\lo(\delta t)$.

We have thus far verified that condition 5) holds, up to \clock{} equivalence, for a ticking clock with a classical register of the cut-off type iff \cref{eq:M 3 term} in the lemma statement holds. By definition of a ticking clock (\cref{def_quantumclock}), we only need to verify that condition 2) [\cref{eq_extension CR to CR}] holds for a ticking clock with a cut-off register. The case \cref{eq_extension CR to CR} is verified analogously to the periodic register case above.
\end{proof}

\subsection{Proof of \cref{prop:autonmous Lindbald form}}\label{sec: proof of autonmous Lindbald form prop}
\ExplicitTickingClockForm*
\begin{proof}
The proposition follows straightforwardly from a more technical representation (\cref{lem_clockgenerators}) discussed in \app{} \cref{Sec:Implicit Ticking Clock Representation}. Specifically, if one expands to leading order in $t$ the channel $\cM_{\cl \to \cl\reT}^{t,k}$ for $k=0,1,\ldots, N_T$ using \cref{eq:clock dynamics landblad}, one finds \cref{eq:M 3 term}. Furthermore, since \cref{eq:clock dynamics landblad} is manifestly Markovian, the expression \cref{eq:lemma 1 stement channel at time} also holds. Since it was established in \cref{lem_clockgenerators}, that this form of the channel is both necessary and sufficient for the channel to be a ticking clock (\cref{def_quantumclock}), we conclude the proof of the proposition.
\end{proof}

\subsection{Proof of \cref{prop:clockwork channel}}\label{sec: proof of clockwork channel prop}
\clockchannel*
\begin{proof}
Consider a ticking clock $(\rho_{\cl\reT}^0, (\cM^{t}_{\cl \reT \to \cl \reT})_{t\geq 0})$ with a classical register given by the following expression for all $t\geq 0$ and $k=0,1,\ldots,N_T$:
\begin{subequations}\label{eq:prop clockwork channel}
	\begin{align}\label{eq:lemma 1 stement channel at time clockchannel}
	\begin{split}
	\cM^{t,k}_{\cl \to \cl \reT}(\rho_\cl^0)
	&=
	\! \lim_{\substack{N\to+\infty \\ N\in\nn}} \! \left( \cM^{t/N}_{\cl\reT \to \cl \reT} \right)^{\!\comP(N-1)}\!\! \circ \cM^{ t/N, k}_{\cl \to \cl \reT}(\rho_\cl^0),
	\end{split}
	\end{align}	
	where
	\begin{align}
	\cM^{t/N,\,k}_{\cl \to \cl \reT}(\cdot)
	= \,&(\cdot) \otimes \proj{k}_\reT + \left(\!\frac{t}{N}\!\right)\mathcal{C}_{(1)}(\cdot) \otimes\proj{k}_\reT
	+ \left(\!\frac{t}{N}\!\right)\mathcal{C}_{(2)}(\cdot) \otimes \proj{k\!+\!1}_\reT
	,\label{eq:M 3 term clockchannel}
	\end{align}
	with
	\begin{align}
	\mathcal{C}_{(1)}(\cdot)&:=  - \mi \,[H, \cdot]  - \sum_{j=1}^\m \frac{1}{2} \{L^{\dagger}_j L_j + J^{\dagger}_j J_j, \cdot\} + L_j (\cdot) L_j^{\dagger},\label{eq:C 1 l clockchannel}\\
	\mathcal{C}_{(2)}(\cdot)&:= \sum_{j=1}^\m J_j (\cdot) J_j^{\dagger},\label{eq:C 2 l clockchannel}
	\end{align}
\end{subequations}
and where $H\in \bounded(\mathcal{H}_\cl)$ is Hermitian and $(L_j)_j$, $(J_j)_j$ are arbitrary operators in  $\bounded(\mathcal{H}_\cl)$ and 
\begin{align}\label{eq:defs of basis cut-off and periodic}
\ket{l}_\reT:= \begin{cases}
\ket{l \text{ mod. } N_T+1}_\reT &\mbox{ for } l\in \nn \text { in the periodic register case.}\\
\ket{N_T}_\reT &\mbox{ for } l= N_T, N_T+1, N_T+2, \ldots \text { in the cut-off register case.}
\end{cases}
\end{align}
To start with, we have to justify that the above channel is a ticking clock. In the case of a periodic register, this is obvious since it is identical to the channel in \cref{lem_clockgenerators}. In the case of the cut-off register, it clearly satisfies condition 5) [\cref{eq:condition 5}]. Thus the only condition remaining to conclude that it is indeed a representation of a ticking clock with a classical cut-off register, is condition 2) [\cref{eq_extension CR to CR}]. There are numerous ways to show this, the most direct is to note by direct substitution that the following Lindbladian is a generator for the channel:
\begin{align}
\mathcal{L}_{\cl\reT}':=  -\mi[\tilde{H},(\cdot)] + \sum_{j=1}^\m  \tilde{L}_{j} (\cdot) \tilde{L}^\dagger_{j} - \sum_{j=1}^\m\frac{1}{2} \big\{ \tilde{L}^\dagger_{j} \tilde{L}_{j}, (\cdot) \big\} + \sum_{l=0}^{N_T}\sum_{j=1}^\m \bar{J}_j^{(l)} (\cdot) \bar{J}_j^{(l)\dagger} - \sum_{l=0}^{N_T}\sum_{j=1}^\m\frac{1}{2} \big\{ \bar{J}_j^{(l)\dagger} \bar{J}_j^{(l)}, (\cdot) \big\},\label{}
\end{align}
where, as before $\tilde H:=H \otimes \id_\reT$, $\tilde L_j:=L_j \otimes \id_\reT$, and the new tick generators are
\begin{align}
\bar J_j^{(l)}:= \begin{cases}
J_j\otimes \ketbra{l+1}{l}_\reT &\mbox{ for } l=0,1,2,\ldots,N_T-1\\
J_j\otimes \ketbra{N_T}{N_T}_\reT &\mbox{ for } l=N_T.
\end{cases}
\end{align}
What is more, even in the cut-off register case, \cref{eq:prop clockwork channel} is the same as that in \cref{lem_clockgenerators} up to \clock{} equivalence. 

Now that we have justified the form of the channel \cref{eq:prop clockwork channel}, we proceed to calculate its \clock{} channel as per the definition \cref{eq:clockwork channel def}:

\begin{align}
\cM^{t}_{\cl \to \cl }(\cdot) &= \lim_{\substack{N\to+\infty \\ N\in\nn}} \tr_\reT\left[ \left( \cM^{t/N}_{\cl\reT \to \cl \reT} \right)^{\!\comP(N-1)}\!\! \circ \cM^{ t/N, k}_{\cl \to \cl \reT}(\cdot)  \right]. \label{eq:M c to c secondary}
\end{align}
Furthmore, observe that one has the following expansion
\begin{align}
\left( \cM^{\delta t}_{\cl\reT \to \cl \reT} \right)^{\!\comP(m-1)}\!\! \circ \cM^{ \delta t, k}_{\cl \to \cl \reT}(\cdot) = \sum_{l=0}^{m} M_m^{(l,\delta t)}(\cdot) \otimes \proj{l+k}_\reT,\label{eq:expensional trial}
\end{align}
for some $M_m^{(l,\delta t)}\in\bounded(\mathcal{H}_\cl)$ which may be $k$-dependent. To see {that} a solution of the form \cref{eq:expensional trial} exists, note that every application of the channel \cref{eq:M 3 term clockchannel} only contains terms which either keep the support of the register the same, i.e. has support on $\proj{k}_\reT$, or increases by one, i.e. has support on $\proj{k+1}_\reT$. Furthermore, the summation ranges from $0$ to $m$ after $m$ applications of the channel, which follows easily inductively.

Hence by comparing \cref{eq:expensional trial,eq:M c to c secondary} one sees that to prove that $\cM^t_{\cl \to\cl}$ is $k$-independent,  it suffices to show that the channels $\Big( M_m^{(l,\delta t)}(\cdot)\Big)_{l=0}^{m}$ are $k$-independent for all $m\in\nnp$, $\delta t\geq 0$.

This is most easily shown by induction. We start by showing that $\Big( M_{1}^{(l,\delta t)}(\cdot)\Big)_{l=0}^{1}$ are $k$-independent by equating
\begin{align}
\cM^{ \delta t, k}_{\cl \to \cl \reT}(\cdot) = \sum_{l=0}^{1} M_1^{(l,\delta t)}(\cdot) \otimes \proj{l+k}_\reT
\end{align}
with \cref{eq:C 1 l clockchannel,eq:C 2 l clockchannel} to find $M_1^{(0,\delta t)}= \idch_\cl + \delta t\,\mathcal{C}_{(1)}(\cdot)$, $M_1^{(1,\delta t)}=  \delta t\,\mathcal{C}_{(2)}(\cdot)$. Therefore, $\Big( M_{1}^{(l,\delta t)}(\cdot)\Big)_{l=0}^{1}$ are $k$-independent since $\mathcal{C}_{(1)}$ and $\mathcal{C}_{(2)}$ are. Now assume  $\Big( M_m^{(l,\delta t)}(\cdot)\Big)_{l=0}^{m}$ are $k$-independent. We show that it follows that 
 $\Big( M_{m+1}^{(l,\delta t)}(\cdot)\Big)_{l=0}^{m+1}$  are $k$-independent:
\begin{align}
&\sum_{l=0}^{m+1} M_{m+1}^{(l,\delta t)}(\cdot) \otimes \proj{l+k}_\reT \label{eq:sum M proof  line 1}\\
& = \left( \cM^{\delta t}_{\cl\reT \to \cl \reT} \right)^{\!\comP m}\!\! \circ \cM^{ \delta t, k}_{\cl \to \cl \reT}(\cdot) \\
&=\cM^{\delta t}_{\cl\reT \to \cl \reT} \circ \left( \sum_{l=0}^{m} M_m^{(l,\delta t)}(\cdot) \otimes \proj{l+k}_\reT  \right)\\
& = \sum_{l=0}^{m} M_1^{(0,\delta t)} \circ M_m^{(l,\delta t)}(\cdot) \otimes \proj{l+k}_\reT  +   M_1^{(1,\delta t)} \circ M_m^{(l,\delta t)}(\cdot) \otimes \proj{l+k+1}_\reT.\label{eq:sum M proof  line 4}
\end{align}
Therefore, since $M_1^{(0,\delta t)}$, $M_1^{(1,\delta t)}$ are manifestly $k$-independent and $M_m^{(l,\delta t)}$ are $k$-independent by assumption, it follows by equating terms in lines \eqref{eq:sum M proof  line 1} and \eqref{eq:sum M proof  line 4}, that  $\Big( M_{m+1}^{(l,\delta t)}(\cdot)\Big)_{l=0}^{m+1}$  are $k$-independent. Hence by induction, we conclude that $\Big( M_{m}^{(l,\delta t)}(\cdot)\Big)_{l=0}^{m}$ are $k$-independent for all $m\in\nnp$ and thus that \cref{eq:M c to c secondary} is $k$-independent also.

Observe that the only difference between the periodic register case and the cut-off register case, are the kets $\ket{l}_\reT$, which are defined in \cref{eq:defs of basis cut-off and periodic}. Since $\tr[\proj{l}_\reT]=1$ for all $l\in\nno$ in both cases, we conclude that \cref{eq:M c to c secondary} is the same in both cases.

Finally, proceeding similarly to the above inductive proof, we observe that
\begin{align}
\cM^{t}_{\cl \to \cl }(\cdot) &= \lim_{\substack{N\to+\infty \\ N\in\nn}} \left( \cM^{t/N}_{\cl \to \cl} \right)^{\!\comP N},
\label{eq:M c to c final}
\end{align}
where 
\begin{align}
\left(\cM^{ t/N}_{\cl \to \cl}\right)^{\!\comP N}(\cdot) = \sum_{l=0}^{N} M_N^{(l,\,t/N)}(\cdot).
\end{align}
It is now straightforward to verify that the above channel constitutes a dynamical semigroup thus admitting a generator representation of the form $\cM^{t}_{\cl \to \cl }(\cdot)=\me^{t \mathcal{L}_\cl}(\cdot)$ with 
\begin{align}
\mathcal{L}_\cl(\cdot)= \lim_{t\to 0^+} \frac{\cM^{t}_{\cl \to \cl } -\idch_{\cl\reT}}{t}   =\lim_{t\to 0^+} \frac{M_1^{(0,t)}+M_1^{(1,t)}-\idch_{\cl\reT}}{t}= \mathcal{C}_{(1)}(\cdot)+\mathcal{C}_{(2)}(\cdot)\,.
\end{align}
Thus $\mathcal{L}_\cl(\cdot)$ is equal to the r.h.s. of \cref{eq:clock dynamics landblad eq} under the replacements $\tilde H\mapsto H$, $\tilde L_j\mapsto L_j$ and $ \tilde J_j^{(l)}\mapsto J_j$.
\end{proof}
 
\subsection{Proofs of \cref{lem: p norm equiv,lem: up bound on m,}}
  
\begin{lemma}[Entry-wise and $p$-norm equivalence]\label{lem: p norm equiv}
Let the complex finite dimensional matrix $A\in\cc^l\times \cc^m$ have entries denoted by $A_{qr}\in\cc$. Let $\| \cdot \|_p$ denote the operator norm on $\cc^l\times \cc^m$ induced by the vector $p$-norm on vector in $\cc^m$. Let $\lo(\delta)$ denote ``little o'' notation for some limit $\delta\to a$. It follows that 
\begin{align}
A_{qr}=\lo(\delta)
\end{align}
for all $q=1,2,3,\ldots l; r=1,2,3,\ldots,m$ if and only if 
\begin{align}
\| A\|_p=\lo(\delta).
\end{align}
The statement holds for any $p>0$.
\end{lemma}
\begin{proof}
	Given the expression for the operator norm, namely
\begin{align}
\| A\|_p= \sup_{v\in\rr^m; \,\|v\|_p\leq 1} \left(\, \sum_{q=1}^l \left|\sum_{r=1}^m A_{qr} v_r \right|^p\,\,\right)^{1/p},
\end{align} 
the direction $A_{qr}=\lo(\delta) \,\,\forall q,r \implies \| A\|_p=\lo(\delta)$ follows easily.  To prove the converse, we will use proof by contradiction.
Suppose $\| A\|_p=\lo(\delta)$, and by contradiction, assume that there exists matrix entry $A_{st}$ s.t. $A_{st}\neq \lo(\delta)$. Therefore, 
\begin{align}\label{eq: A coe contra}
\lim_{\delta\to a} \frac{|A_{st}|}{\delta}>0.
\end{align}
We can now use the definition of the operator norm to achieve the lower bound
\begin{align}
\| A\|_p \geq  \left(\, \sum_{q=1}^l \left|\sum_{r=1}^m A_{qr} \delta_{r,t} \right|^p\,\,\right)^{1/p}=  \left(\, \sum_{q=1}^l \left|A_{qt} \right|^p\,\,\right)^{1/p} \geq |A_{st}|.
\end{align}
Therefore, dividing both sides by $\delta$ followed by taking the limit $\delta\to a$ we achieve using \cref{eq: A coe contra} that $\lim_{\delta\to a} \frac{\|A\|_p}{\delta}>0$. This contradicts the assertion that $\| A\|_p=\lo(\delta)$.
\end{proof}

\begin{lemma}[Maximum number of Lindblad operators needed]\label{lem: up bound on m} Consider a \clock{} of Hilbert space dimension $d\in\nnp$. For every Hermitian operator $H$ and two finite sequences of operators $(L_j)_{j=1}^\m$, $(J_j)_{j=1}^\m$ on $\bounded(\mathcal{H}_\cl)$ giving rise to the channel $\cM^{t,k}_{\cl \to \cl \reT}(\cdot)$ via \cref{eq:lemma 1 stement channel at time}; there exits $2 (d^2-1)$ new operators $(L_j')_{j=1}^{d^2-1}$, $(J_j')_{j=1}^{d^2-1}$ on $\bounded(\mathcal{H}_\cl)$ such that the channel $\cM^{t,k}_{\cl \to \cl \reT}(\cdot)$ is invariant under the mappings
\begin{align}
\sum_{j=1}^\m -\frac{1}{2} \{L^{\dagger}_j L_j + \theta(k)J^{\dagger}_j J_j, \cdot\} + L_j (\cdot) L_j^{\dagger}\, &\mapsto \, \sum_{j=1}^{d^2-1} -\frac{1}{2} \{L^{\prime\dagger}_j L_j' + \theta(k)J^{\prime\dagger}_j J_j', \cdot\} 
+ L_j' (\cdot) L_j^{\prime\dagger}\\
  \sum_{j=1}^\m J_j (\cdot) J_j^{\dagger}\, &\mapsto \, \sum_{j=1}^{d^2-1} J_j' (\cdot) J_j^{\prime\dagger}
\end{align}	
in \cref{eq:M 3 term,eq:C 2 l} respectively.	
\end{lemma}
\begin{proof}
In the proof of \cref{lem_clockgenerators}, $\m$ is simply a non negative integer arising from writing an arbitrary implementation of the channel $\cM^{t,k}_{\cl \to \cl \reT}(\cdot)$ is Kraus form. To prove \cref{lem: up bound on m}, it will suffice to prove that without loss of generality, $\m$ in the proof of \cref{lem_clockgenerators} can be chosen to be equal to $d^2-1$. To do so, we start by recalling \cref{eq:classical reg def 2} in the proof of \cref{lem_clockgenerators}:
\begin{align}\label{eq:classical reg def 33}
\cM_{\cl  \to \cl  \reT}^{t,k} (\rho_{\cl }) = \sum_{l=0}^{N_T}\sum_{j=0}^{\NQ } N_j^{(k)}(l,t) \rho_\cl\, {N_j^{(k)}}^\dag(l,t)\otimes\ketbra{l}{l}_\reT,  
\end{align}	
for all $ k=0,1,\ldots,N_T;\, t\geq 0$. Recall also that
$N_j^{(k)}(l,t):= \bra{l}_\reT Q_j^{(k)}(t): \bounded(\mathcal{H}_\cl)\to\bounded(\mathcal{H}_\cl)$ where $Q_j^{(k)}(t): \bounded(\mathcal{H}_\cl)\to\bounded(\mathcal{H}_\cl\otimes\mathcal{H}_\reT)$ and
\begin{align}\label{eq:Q s equation}
\sum_{j=0}^\m {Q_j^{(k)}(t)}^\dag Q_j^{(k)}(t)=\id_\cl.
\end{align} 
Using the resolution of the identity, \cref{eq:Q s equation} implies 
\begin{align}\label{eq:N s equation}
\sum_{l=0}^{N_T} \sum_{j=0}^\m {N_j^{(k)}(l,t)}^\dag N_j^{(k)}(l,t)=\id_\cl.
\end{align} 
Thus since the basis $\{\ket{l}_\reT\}_{l=0}^{N_T}$ is orthogonal, the operators $\{{N_j^{(k)}(l,t)}\}_{j=0}^{\m}$ are completely arbitrary and independent from $\{{N_j^{(k)}(l',t)}\}_{j=0}^{\m}$ for all $l\neq l'$, $l=0,1,2,\ldots, N_T$; up to the normalisation imposed by \cref{eq:N s equation}. The lemma now follows directly from Choi's theorem~\cite{Choi}. To see this, note that the channel
\begin{align}\label{eq:N N dag l t}
\sum_{j=0}^{\NQ } N_j^{(k)}(l,t) \,(\cdot)\, {N_j^{(k)}}^\dag(l,t) : \bounded(\mathcal{H}_\cl) \to \bounded(\mathcal{H}_\cl),
\end{align}
is completely positive for all $\m\in\nnp$. Therefore, via Choi's theorem there exists operators $(N_j^{(k)\prime})_{j=1}^{d^2-1}$ such that 
\begin{align}\label{eq:N N dag l t 3}
\sum_{j=0}^{\NQ } N_j^{(k)}(l,t) \,(\cdot)\, {N_j^{(k)}}^\dag(l,t) = \sum_{j=0}^{d^2-1} N_j^{(k)\prime}(l,t) \,(\cdot)\, {N_j^{(k)\prime}}^\dag(l,t),
\end{align}
where 
\begin{align}\label{eq:N s equation 2}
\sum_{l=0}^{N_T} \sum_{j=0}^{d^2-1} {N_j^{(k)\prime}(l,t)}^\dag N_j^{(k)\prime}(l,t)=\id_\cl.
\end{align}
However, since the operators $\big(N_j^{(k)}\big)_{j=1}^{N_L}$ were arbitrary to begin with (up to the aforementioned normalisation which the operators $\big(N_j^{(k)\prime}(l,t)\big)_{j=1}^{d^2-1}$ also satisfy), we can always choose $N_L=d^2-1$ from the outset.
\end{proof}



\end{appendices}

\end{document}